\documentclass{article}
\usepackage[T1]{fontenc}
\usepackage[margin=1in]{geometry}
\usepackage{amsfonts,amsmath,amsthm,amssymb,mathtools}  %

\mathtoolsset{centercolon}
\usepackage{xfrac,nicefrac}
\usepackage{mathdots}
\usepackage{mleftright}  %

\usepackage{todonotes}

\usepackage{xspace}
\xspaceaddexceptions{]\}}  %
\usepackage{regexpatch}

\usepackage{bm,bbm,dsfont}  %
\usepackage{caption}
\usepackage[normalem]{ulem}
\usepackage{enumitem}

\usepackage{graphicx}
\usepackage{float}
\usepackage{subcaption}  %
\usepackage{tcolorbox}
\usepackage{tikz}
\usetikzlibrary{decorations.pathreplacing}
\usetikzlibrary{calc}
\usetikzlibrary{positioning}
\usetikzlibrary{arrows.meta}

\usepackage[linesnumbered,boxed,ruled,vlined]{algorithm2e}

\SetCommentSty{mycommfont}
\SetKwComment{Comment}{/* }{ */}

\usepackage[noend]{algpseudocode}

\usepackage{thmtools,thm-restate}
\usepackage[colorlinks,citecolor=blue,linkcolor=blue,urlcolor=red]{hyperref}

\theoremstyle{plain}

\newtheorem{theorem}{Theorem}[section]  %
\newtheorem{lemma}[theorem]{Lemma}

\newtheorem{corollary}[theorem]{Corollary}

\newtheorem{claim}[theorem]{Claim}

\theoremstyle{definition}  %

\newtheorem{property}[theorem]{Property}

\usepackage[capitalise]{cleveref}
\crefname{algocf}{Algorithm}{Algorithms}
\crefname{claim}{Claim}{Claims}
\crefname{property}{Property}{Properties}
\Crefname{prop}{Property}{Properties}
\Crefname{lemma}{Lemma}{Lemmata}

\usepackage{xurl}

\newfloat{Distribution}{htbp}{loa}
\crefname{Distribution}{Distribution}{Distributions}
\Crefname{Distribution}{Distribution}{Distributions}
\newfloat{Protocol}{htbp}{loa}
\crefname{Protocol}{Protocol}{Protocols}
\Crefname{Protocol}{Protocol}{Protocols}

\SetKwProg{Function}{Function}{:}{}

\SetKwProg{CodeBlock}{}{}{}
\SetKwProg{Repeat}{Repeat}{:}{}

\DeclarePairedDelimiter{\bk}{(}{)}
\DeclarePairedDelimiter{\Bk}{[}{]}

\DeclarePairedDelimiterX\mysetbase[2]{\lbrace}{\rbrace}{#1\,\delimsize\vert\,#2}
\NewDocumentCommand{\myset}{sO{}m m}{%
  \IfBooleanTF{#1}%
    {\mysetbase*{#3}{#4}}%
    {\mysetbase[#2]{#3}{#4}}%
}

\DeclareMathOperator*{\E}{\mathbb{E}}

\let\Pr\PrAux
\DeclareMathOperator{\poly}{poly}

\renewcommand{\tilde}{\widetilde}

\newcommand{\defeq}{\coloneqq}
\newcommand{\eps}{\varepsilon}

\renewcommand{\epsilon}{\eps}

\newcommand{\defn}[1]{\emph{\boldmath\textbf{#1}}}

\usepackage{regexpatch}
\makeatletter
\xpatchcmd\thmt@restatable{%
\csname #2\@xa\endcsname\ifx\@nx#1\@nx\else[{#1}]\fi
}{%
\ifthmt@thisistheone
\csname #2\@xa\endcsname\ifx\@nx#1\@nx\else[{#1}]\fi
\else
\csname #2\@xa\endcsname[{Restated}]
\fi}{}{}
\makeatother

\newfloat{Algorithm}{htbp}{loa}

\newcommand{\buf}{\textsf{Buf}}
\newcommand{\dbuf}{\textsf{dBuf}}

\newcommand{\algAdv}{\textsf{Alg}_{\text{adv}}}
\newcommand{\algBase}{\textsf{Alg}_{\text{base}}}
\newcommand{\algAdapt}{\textsf{Alg}_{\text{adpt}}}
\newcommand{\algAdaptMerge}{\textsf{Alg}_{\text{adpt+mrg}}}
\newcommand{\algFinal}{\textsf{Alg}_{\text{final}}}
\newcommand{\algPool}{\textsf{Alg}_{\text{pool}}}

\title{Nearly Optimal Bounds for Stochastic Online Sorting}

\author{Yang Hu\thanks{Institute for Interdisciplinary Information Sciences, Tsinghua University. \texttt{y-hu22@mails.tsinghua.edu.cn}. Work performed while the author was visiting Carnegie Mellon University.}}
\date{}

\begin{document}

\maketitle

\begin{abstract}
    In the \emph{online sorting} problem, we have an array $A$ of $n$ cells, and receive a stream of $n$ items $x_1,\dots,x_n\in [0,1]$. When an item arrives, we need to immediately and irrevocably place it into an empty cell. The goal is to minimize the sum of absolute differences between adjacent items, which is called the \emph{cost} of the algorithm. It has been shown by Aamand, Abrahamsen, Beretta, and Kleist (SODA 2023) that when the stream $x_1,\dots,x_n$ is generated adversarially, the optimal cost bound for any deterministic algorithm is $\Theta(\sqrt{n})$.

In this paper, we study the \emph{stochastic} version of online sorting, where the input items $x_1,\dots,x_n$ are sampled uniformly at random. Despite the intuition that the stochastic version should yield much better cost bounds, the previous best algorithm for stochastic online sorting by Abrahamsen, Bercea, Beretta, Klausen and Kozma (ESA 2024) only achieves $\tilde{O}(n^{1/4})$ cost, which seems far from optimal. We show that stochastic online sorting indeed allows for much more efficient algorithms, by presenting an algorithm that achieves expected cost $\log n\cdot 2^{O(\log^* n)}$. We also prove a cost lower bound of $\Omega(\log n)$, thus show that our algorithm is nearly optimal.
\end{abstract}

\section{Introduction}

The \defn{online sorting problem} \cite{Aamand23,Abrahamsen2024OnlineSA} captures the following basic question about placing elements into an array so that their order is as close to sorted as possible. We have an empty array $A$ with $n$ \defn{cells}, and receive a stream of $n$ real numbers $x_1,\dots,x_n\in [0,1]$ (henceforth called \defn{items}). When an item arrives, it must be immediately and irrevocably placed into an empty cell. The goal is to minimize $\sum_{i=1}^{n-1}|A[i+1]-A[i]|$, the sum of differences between adjacent items. This sum is called the \defn{cost} of the algorithm.\footnote{A closely related metric is the \defn{competitive ratio} of the algorithm, which is the ratio between the algorithm's cost and the optimal offline cost. Previous work has added the (essentially wlog) assumption that at least one item is $0$ and at least one is $1$, so that these two metrics are equivalent. In the setting that we will be focusing on in this paper (the stochastic version of the problem), this wlog assumption is not even necessary, since even without it the two metrics are within a constant factor of each other with high probability.}

The online sorting problem was first introduced by Aamand, Abrahamsen, Beretta, and Kleist \cite{Aamand23}, who gave a simple deterministic algorithm achieving cost $O(\sqrt{n})$, and proved a matching lower bound for any deterministic algorithm. Subsequent work by Abrahamsen, Bercea, Beretta, Klausen and Kozma \cite{Abrahamsen2024OnlineSA} showed that this bound of $\Theta(\sqrt{n})$ is the optimal not just for deterministic algorithms but also for randomized ones. 

\paragraph{Stochastic online sorting.} A natural question is what happens if the stream $x_1, x_2, \ldots, x_n$ is generated non-adversarially so that each $x_i$ is i.i.d.~uniformly random. Intuitively, this \defn{stochastic} version of the problem allows for a much better (expected) cost bound. For example, if we imagine that the $i$-th to final insertion is able to achieve a cost of $O(1/i)$, which intuitively would make sense since it is able to choose from a collection of $i$ remaining positions, then the total cost bound would end up being $O(\log n)$. 

Surprisingly, initial work on the stochastic setting \cite{Abrahamsen2024OnlineSA} has yielded only relatively minor improvements over the worst-case setting. In a simple but elegant algorithm, Abrahamsen et al. \cite{Abrahamsen2024OnlineSA} showed how to improve the cost bound to $\tilde{O}(n^{1/4})$. They leave open the question of proving any nontrivial lower bound, as well as the question of whether significantly better bound than $\tilde{O}(n^{1/4})$ might be possible.\footnote{This problem has also been discussed at length in open problem sessions at several recent data-structure workshops.}

In this paper we seek to answer the following basic question. Is it really the case that the stochastic online sorting admits bounds that are only slightly better than for the worst-case version of the problem? Or, is it somehow possible to obtain an algorithm that achieves much better bounds on the cost? 

\paragraph{The relationship between stochastic online sorting and hashing.} A natural attempt to solve stochastic online sorting is to apply linear probing \cite{knuth63linprobe}, which is a classical algorithm for hash tables. That is, when an item $x_i$ arrives, we first attempt to insert it in the cell $h(x_i)=\lceil x_i\cdot n\rceil$. If this cell is already full, we will insert $x_i$ into the first empty cell in $h(x_i)+1,h(x_i)+2,\dots$ (wrapping around if necessary). Intuitively, in this algorithm, a cell $i$ should contain value close to $i/n$, which means that the total cost should be small. Indeed, \cite{Abrahamsen2024OnlineSA} shows that, if the number of cells is slightly larger than the number of items (say, $\gamma\cdot n$ for some $\gamma>1$), then the expected cost of stochastic online sorting can be improved to $O(1+1/(\gamma-1))$ with this approach. However, for the case where the number of cells is exactly $n$ (which is the case that this paper is interested in), this approach ends up being (perhaps surprisingly) inefficient.

What prevents us from getting good bounds is that, linear probing famously suffers from a clustering effect. That is, the inserted items tend to cluster together into long runs, so that new items $x_i$ may be inserted very far from $h(x_i)$. In the context of our problem, this clustering effect causes such a substantial blowup in the cost that linear probing, on its own, ends up with a cost bound of roughly $O(\sqrt{n})$, which is much worse than the $O(\log n)$-style bound that we might have originally hoped for.

Thus, one can view the stochastic online sorting problem, from a high level, as being a variation on the classical problem of how to place elements into a hash table: Specifically, how would you place the elements if you wished to minimize the sums of the differences between the \emph{hashes} of consecutive elements? Is linear probing, or something like linear probing optimal? Or is it possible to somehow bypass the clustering bottleneck that arises in linear probing in order to achieve a much better (and perhaps even sub-polynomial) cost bound?

\paragraph{This paper: Nearly tight upper and lower bounds.}

The goal of this paper is to determine the optimal cost of stochastic online sorting. Our main result is that, it is indeed possible to achieve much better bounds than $\poly n$. In fact, we can achieve cost close to $\log n$, which is consistent with previous intuition. Along the way, we develop new techniques that fully exploit the online nature of the problem, which may find applications to other similar problems. 

Our main result is the following theorem:

\begin{restatable}{theorem}{mainTheorem}
    \label{thm:real_exp}
    There exists a deterministic algorithm for stochastic online sorting, that achieves expected cost $\log n\cdot 2^{O(\log^* n)}$.
\end{restatable}

We remark that the randomization in Theorem \ref{thm:real_exp} comes entirely from the fact that the input itself is random. The algorithm, on the other hand, is entirely deterministic.

Our second result is a nearly matching lower bound of $\Omega(\log n)$ on the expected cost that any stochastic online sorting algorithm must incur. This lower bound applies to both deterministic and randomized algorithms.

\begin{restatable}{theorem}{mainLB}
    \label{thm:main_lb}
    Any randomized algorithm for stochastic online sorting must have expected cost at least $\Omega(\log n)$.
\end{restatable}

Combined, Theorems \ref{thm:real_exp} and \ref{thm:main_lb} come \emph{remarkably close} to determining the optimal cost bound for the stochastic online-sorting problem. Interestingly, the gap of $2^{\Omega(\log^* n)}$ appears to be genuine for the algorithm that we use to prove Theorem \ref{thm:real_exp}---it is not just an artifact of the analysis.

In addition to considering \emph{expected} cost, we also consider \emph{high-probability} cost. That is, what is the best bound on cost that an algorithm can hope to achieve with high probability in $n$? Here, we show that a $\poly\log n$ cost bound is possible: 

\begin{restatable}{theorem}{highProb}
    \label{thm:warm_up_2_high_prob}
    There exists a deterministic algorithm for stochastic online sorting that, with probability $1-O(1/n^{10})$, achieves cost $\poly\log n$.
\end{restatable}

An interesting feature of Theorem \ref{thm:warm_up_2_high_prob} is that, unlike the other algorithms in this paper (and the stochastic sorting algorithm in \cite{Abrahamsen2024OnlineSA}), the algorithm for Theorem \ref{thm:warm_up_2_high_prob} is entirely non-recursive. Instead, it is able to allocate regions of the array, one after another, to different sub-intervals of $[0, 1]$, in order to decide where to send elements as they arrive. The allocations themselves are performed at a relatively fine granularity, and rely on two algorithmic techniques that we develop (the notions of adaptive allocation and segment synchronization). 

The rest of the paper proceeds as follows. After presenting the related work (Section \ref{sec:related}) and the technical overview (Section \ref{sec:overview}), we present the formal algorithms and analyses in \cref{sec:warm_up,sec:warm_up_2,sec:final}. We first present the two main techniques used in the algorithm: \defn{adaptive allocation} (\cref{sec:warm_up}) and \defn{segment synchronization} (\cref{sec:warm_up_2}). In each of these sections, we use the techniques to design a weaker version of the final algorithm, where $\algAdapt$ (presented in \cref{sec:warm_up}) achieves $n^{o(1)}$ cost, and $\algAdaptMerge$ (presented in \cref{sec:warm_up_2}) achieves $\poly\log n$ cost \emph{with high probability}. After these warmups, we obtain our final algorithm in \cref{sec:final}, where we show how to impose a recursive structure on $\algAdaptMerge$ so that we can reduce the cost to $\log n\cdot 2^{O(\log^*n)}$. Finally, in \cref{sec:lb}, we prove the lower bound of $\Omega(\log n)$, and in \cref{sec:open_problems}, we present some open questions.

\section{Related Work}
\label{sec:related}

\paragraph{The online sorting problem.} The online sorting problem was first introduced in \cite{Aamand23}, where they proved the $\Theta(\sqrt{n})$ cost bound for deterministic algorithms on adversarial inputs, and used it to prove lower bounds for many geometric online packing problems. \cite{Abrahamsen2024OnlineSA} further generalized the $\Omega(\sqrt{n})$ cost lower bound to randomized algorithms. 

\paragraph{Variants of online sorting.} \cite{Aamand23} also studied the setting where the number of cells is more than the number of items, and showed that the cost bound is significantly better than the standard setting. However, there remains a big gap between their upper and lower bounds in this case. Specifically, for the case where there are $n$ items and $\lceil\gamma\cdot n\rceil$ cells ($\gamma>1$), \cite{Aamand23} gave a cost upper bound of $2^{O(\sqrt{\log n}\cdot \sqrt{ \log\log n+\log(1/(\gamma-1)}))}$, and a lower bound of $\Omega(\log n/(\log\log n\cdot \gamma))$.

For the stochastic version of online sorting, \cite{Abrahamsen2024OnlineSA} proved a $\tilde{O}(n^{1/4})$ cost upper bound, as well as an upper bound of $O(1+1/(\gamma-1))$ for the stochastic setting where there are $n$ items and $\lceil\gamma\cdot n\rceil$ cells ($\gamma >1$).

For online sorting with different metrics (i.e., the items are not real numbers, but objects from a different metric space), \cite{Abrahamsen2024OnlineSA} showed that: For the uniform metric (where $d(x,y)=[x\ne y]$), there exists an algorithm achieving competitive ratio $O(\log n)$, and for the multi-dimensional euclidean metric (i.e., items are from $\mathbb{R}^d$) when the dimension $d$ is a constant, a competitive ratio of $\tilde{O}(\sqrt{n})$ can be achieved.

\paragraph{Related problems.} Online sorting is similar to many famous online problems, and we present two of them here. It is an interesting research direction to study the relationship between these problems.

\paragraph{List labeling.} In list labeling \cite{itai81sparse,dietz82maintaining,dietz87two,saks18online,bender22online,bender24nearly}, similar to online sorting, we receive a sequence of items arriving in an online fashion, and are asked to place the items in sorted order. However, in this problem we are allowed to move items arbitrarily after they are inserted, and the goal is to minimize \emph{recourse} (i.e., the number of items moved) while guaranteeing that the items are \emph{always} in sorted order. The most common setting is when the number of cells is linear in the number of items $n$, i.e. $(1+\Theta(1))n$. In this case, the best algorithm is by \cite{bender24nearly}, which achieves amortized recourse of $O(\log n\cdot (\log\log n)^3)$.

\paragraph{Matching on the line.} Matching on the line \cite{gupta2012online,raghvendra2018optimal,gupta2019stochastic,balkanski2023power} is a special case of online matching, that is very similar to online sorting. The main difference between the two problems lies in the objective function: for matching on the line, each cell also has a value, and the goal is to minimize the sum of differences between the value of each item and the cell that it is placed in.

\paragraph{Independent work.} In a recent parallel and independent work, Kalavas, Platanos and Tolias \cite{kalavas2025polylogarithmiccompetitivealgorithmstochastic} presents an algorithm for stochastic online sorting, that achieves $O(\log^2 n)$ cost with high probability.

\section{Technical Overview}
\label{sec:overview}

Our starting point is the algorithm $\algBase$ from \cite{Abrahamsen2024OnlineSA}, which achieves expected cost $\tilde{O}(n^{1/4})$. Along the way, we develop two important techniques for stochastic online sorting: \emph{adaptive allocation} and \emph{segment synchronization}.

\subsection{The Base Algorithm: Predict the Insertions}

The main observation by \cite{Abrahamsen2024OnlineSA} is that, since the input sequence is stochastic, we can \emph{predict} (using a straightforward Chernoff bound) the number of items that belong to a value interval $[L,R]\subset [0,1]$, and use this prediction to build a good algorithm. We will be using the algorithm $\algAdv$ from \cite{Aamand23} as a black box, which is the algorithm that achieves cost $O(\sqrt{n})$ for \emph{any} sequence of input items.

\paragraph{The algorithm $\algBase$.} $\algBase$ is a recursive algorithm. We first present a non-recursive version of it, then show how to improve it by recursing. Let $B$ be a parameter. We partition the value interval $[0,1]$ into $B$ intervals, each of length $1/B$. These intervals are called \defn{segments}. Using a Chernoff bound, we have that for any segment, with high probability, at least $m=(n/B)-\tilde{O}(\sqrt{n/B})$ of all $n$ items belong to it.

Given this knowledge, we allocate $m$ cells for each segment in advance, so that adjacent cells will contain items that are close to each other. Formally, we partition the $n$ cells into $B+1$ subarrays $\buf_1,\dots,\buf_B,P$, where each $\buf_i$ contains $m$ cells, and only accepts items that belong to segment $i$. We call $\buf_i$ a \defn{buffer} for segment $i$. The last subarray $P$ contains the remaining $n-B\cdot m=\tilde{O}(\sqrt{nB})$ cells, and accepts the items that overflow from the buffers. We call $P$ the \defn{pool}.

The items are processed as follows: Upon receiving item $x$ that belongs to segment $i$, if $\buf_i$ is not filled, then $x$ is inserted in $\buf_i$, otherwise $x$ is inserted in $P$. Next, within the subarray, $\algAdv$ is called to choose an empty cell for the new item.

The cost of this algorithm consists of three parts:
\begin{enumerate}
    \item \textbf{The cost of the buffers:} Fix one buffer, and consider the sum of adjacent differences in it. Since all items in a buffer come from the same segment, we can think of each buffer as a sub-problem with $m$ cells, where each value is from a smaller range $[(i-1)/B,i/B]$ ($i$ is the index of the segment). The cost of such a sub-problem is $O(\sqrt{m}/B)$, which is the cost of $\algAdv$ divided by the value range. Since we have $B$ buffers, their total cost is $O(\sqrt{m})=\tilde{O}(\sqrt{n/B})$.
    \item \textbf{The cost of the pool:} For the pool, the sum of adjacent differences is $O(\sqrt{|P|})=\tilde{O}((nB)^{1/4})$.
    \item \textbf{The inter-subarray cost:} We also need to consider the differences between adjacent subarrays. For this sum, the pool only contributes $O(1)$. As for adjacent buffers, we have that an item from $\buf_i$ and an item from $\buf_{i+1}$ differ by at most $O(1/B)$. Therefore, the total cost for this part is $O(1)$.
\end{enumerate}
Setting $B$ optimally leads to a competitive ratio of $\tilde{O}(n^{1/3})$.

\paragraph{Adding recursions to $\algBase$.} To improve the cost bound, \cite{Abrahamsen2024OnlineSA} further observes that, the items inserted to each buffer $\buf_i$ are also uniformly random (conditioned on the fact that each item belongs to segment $i$). That is, the subproblem of each buffer is actually an instance of \emph{stochastic} online sorting. We therefore can recursively apply the above algorithm on the buffers, so that the \textbf{(recursive) cost of buffers} becomes $\tilde{O}(m^{1/3})$ instead of $O(\sqrt{m})$. By resetting the value of $B$, this results in an algorithm with better cost. We can then reapply the new algorithm on the buffers to further improve the cost. Taken to the limit, the cost of the optimal recursive algorithm $\algBase$ is $\tilde{O}(n^{1/4})$.

\subsection{Adaptive Allocation, or: The Pool is Not Adversarial}

We now present the first improvement to the algorithm, which achieves cost $n^{o(1)}$.

\paragraph{The randomness of the pool.} Our improvement works by observing that, the distribution of insertions to the pool has certain randomness that we can exploit. 

The base algorithm $\algBase$ used the fact that the buffers receive uniformly random items. However, $\algBase$ did not use any nontrivial property about the distribution of items inserted into the pool. Instead, the pool is only handled using $\algAdv$, the algorithm for adversarial inputs. This is doing a great injustice to the pool, since the items that it receives actually follows a nice distribution (albeit not uniformly random).

\paragraph{Adaptive allocation.} It turns out that, we can indeed exploit the randomness of the pool, by leveraging the online nature of our problem. We first imagine that, if the pool receives perfectly uniformly random items, then we could apply a Chernoff bound to lower bound the number of items from each segment to the pool, much like in $\algBase$. We can then allocate new buffers within the pool. Sadly, because different buffers will overflow by significantly different amounts, we cannot say much, a priori, about how many items each segment will end up sending to the pool.

To get around this issue, we observe that we only need to start thinking about the pool when some buffer finally becomes full. This means that we can take a ``lazy'' approach, in which we defer making any predictions about the pool, until the specific moment in time when the pool starts to be used, at which point we can argue that we actually \emph{can} make reasonably good predictions about how many elements each segment will send to the pool in the future, which can guide the allocation of new buffers.

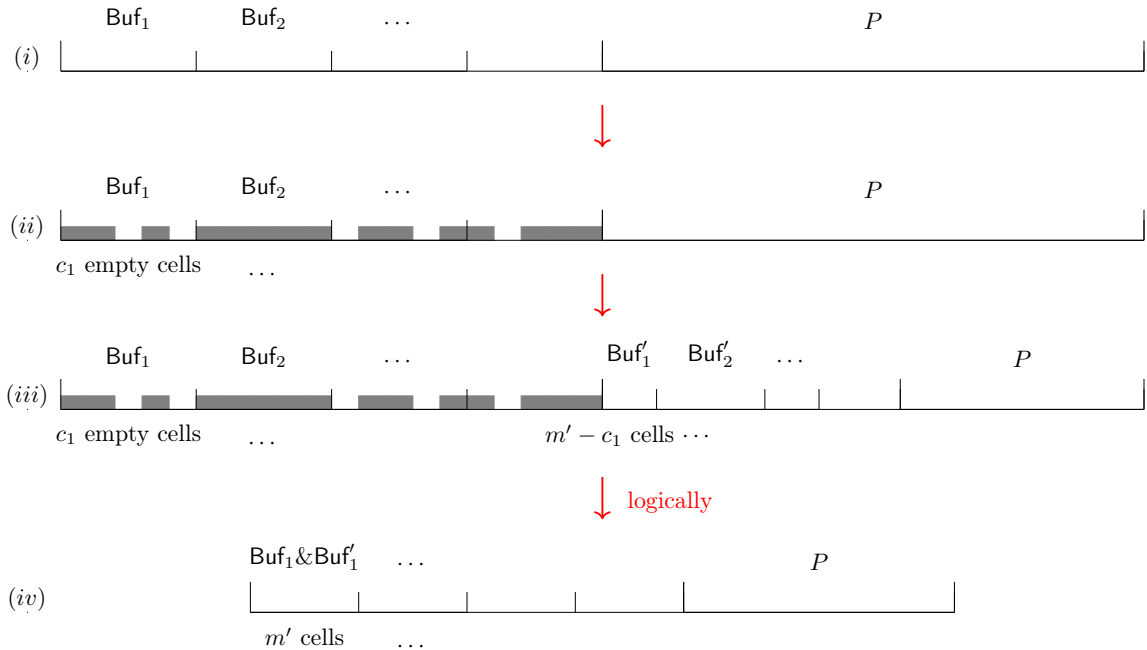
\begin{figure}[ht]
    
    \centering
    \scalebox{0.9}{
    \begin{tikzpicture}
    \def \leftspace {1};
    \def \stepheight {-2.5};
    \def \mytick {0.3};
    \def \drawtick[#1,#2]{\draw ({#1}, {#2}) -- ++ (0,\mytick);} %
    \def \drawlargetick[#1,#2]{\draw ({#1}, {#2}) -- ++ (0,1.5*\mytick);}
    \def \midpoint {\leftspace+8};
    \def \drawsegment[#1,#2,#3,#4,#5,#6,#7,#8]{ %
        \draw ({(#1)}, {(#2)}) -- ({(#1) + (#3)}, {(#2)})
         node[midway, above = 0.1] (#7) {#5}
         node[midway, below = 0.1] (#8) {#6};
        \drawlargetick[(#1),(#2)]
        \drawlargetick[(#1) + (#3),(#2)]
        \foreach \i in {1,...,#4}
        {
            \drawtick[(#1) + \i / (#4 + 1) * (#3),(#2)]
        }
    }

    \draw (\leftspace-0.5, 3*\stepheight) -- ++ (0, 0)
        node[midway, above = -0.1] (1) {$(i)$};
    \drawsegment[\leftspace,3*\stepheight,8,3,,,upper_final,lower_final];
    \drawsegment[\leftspace+8,3*\stepheight,8,0,,,upper_final,lower_final];

    \draw (\leftspace, 3*\stepheight) -- ++ (2*\leftspace, 0)
        node[midway, above = 0.5] (1) {$\buf_{1}$};

    \draw (\leftspace+2, 3*\stepheight) -- ++ (2*\leftspace, 0)
        node[midway, above = 0.5] (1) {$\buf_{2}$};

    \draw (\leftspace+4, 3*\stepheight) -- ++ (2*\leftspace, 0)
        node[midway, above = 0.5] (1) {$\cdots$};

    \draw (\leftspace+8, 3*\stepheight) -- ++ (8*\leftspace, 0)
        node[midway, above = 0.5] (1) {$P$};

    \draw [->, color=red, thick] (\midpoint, 3.2 * \stepheight) -- ++ (0, 0.25 * \stepheight);

    \def \lenarray {{{1,1,0,1,0},{1,1,1,1,1},{0,1,1,0,1},{1,0,1,1,1}}}

    \foreach \i in {0,...,3}
    {
        \foreach \j in {0,...,4}
        {
        
            \filldraw [fill=gray, draw=gray] (\leftspace+2*\i+0.4*\j,4*\stepheight) rectangle ++(0.4*\lenarray[\i][\j],0.2*\lenarray[\i][\j]);
        }
    }

    \draw (\leftspace-0.5, 4*\stepheight) -- ++ (0, 0)
        node[midway, above = -0.1] (1) {$(ii)$};
    \drawsegment[\leftspace,4*\stepheight,8,3,,,upper_final,lower_final];
    \drawsegment[\leftspace+8,4*\stepheight,8,0,,,upper_final,lower_final];

    \draw (\leftspace, 4*\stepheight) -- ++ (2*\leftspace, 0)
        node[midway, above = 0.5] (1) {$\buf_{1}$}
        node[midway, below = 0.1] (1) {$c_1$ empty cells};

    \draw (\leftspace+2, 4*\stepheight) -- ++ (2*\leftspace, 0)
        node[midway, above = 0.5] (1) {$\buf_{2}$}
        node[midway, below = 0.3] (1) {$\cdots$};

    \draw (\leftspace+4, 4*\stepheight) -- ++ (2*\leftspace, 0)
        node[midway, above = 0.5] (1) {$\cdots$};

    \draw (\leftspace+8, 4*\stepheight) -- ++ (8*\leftspace, 0)
        node[midway, above = 0.5] (1) {$P$};

    \draw [->, color=red, thick] (\midpoint, 4.2 * \stepheight) -- ++ (0, 0.25 * \stepheight);

    \foreach \i in {0,...,3}
    {
        \foreach \j in {0,...,4}
        {
        
            \filldraw [fill=gray, draw=gray] (\leftspace+2*\i+0.4*\j,5*\stepheight) rectangle ++(0.4*\lenarray[\i][\j],0.2*\lenarray[\i][\j]);
        }
    }

    \def \newarray {{0,2,6,8,11}}
    
    \foreach \i in {1,...,4}
    {
        \drawtick[\leftspace+8+0.4*\newarray[\i],5*\stepheight]
    }
    \drawlargetick[\leftspace+8+0.4*\newarray[4],5*\stepheight]
    
    \drawsegment[\leftspace,5*\stepheight,8,3,,,upper_final,lower_final];
    \drawsegment[\leftspace+8,5*\stepheight,8,0,,,upper_final,lower_final];

    \draw (\leftspace-0.5, 5*\stepheight) -- ++ (0, 0)
        node[midway, above = -0.1] (1) {$(iii)$};

    \draw (\leftspace, 5*\stepheight) -- ++ (2*\leftspace, 0)
        node[midway, above = 0.5] (1) {$\buf_{1}$}
        node[midway, below = 0.1] (1) {$c_1$ empty cells};

    \draw (\leftspace+2, 5*\stepheight) -- ++ (2*\leftspace, 0)
        node[midway, above = 0.5] (1) {$\buf_{2}$}
        node[midway, below = 0.3] (1) {$\cdots$};

    \draw (\leftspace+4, 5*\stepheight) -- ++ (2*\leftspace, 0)
        node[midway, above = 0.5] (1) {$\cdots$};

    \draw (\leftspace+8+0.4*\newarray[0], 5*\stepheight) -- (\leftspace+8+0.4*\newarray[1],5*\stepheight)
        node[midway, above = 0.5] (1) {$\buf'_{1}$}
        node[midway, below = 0.1] (1) {$m'-c_1$ cells $\cdots$};
        
    \draw (\leftspace+8+0.4*\newarray[1], 5*\stepheight) -- (\leftspace+8+0.4*\newarray[2],5*\stepheight)
        node[midway, above = 0.5] (1) {$\buf'_{2}$};
        
    \draw (\leftspace+8+0.4*\newarray[2], 5*\stepheight) -- (\leftspace+8+0.4*\newarray[3],5*\stepheight)
        node[midway, above = 0.5] (1) {$\cdots$};

    \draw (\leftspace+8+0.4*\newarray[4], 5*\stepheight) -- (\leftspace+16,5*\stepheight)
        node[midway, above = 0.5] (1) {$P$};

    \draw (\leftspace-0.5, 6.2*\stepheight) -- ++ (0, 0)
        node[midway, above = -0.1] (1) {$(iv)$};

    \draw [->, color=red, thick] (\midpoint, 5.4 * \stepheight) -- ++ (0, 0.25 * \stepheight);

    \node[color=red] at (\midpoint+1, 5.55 * \stepheight) {logically};
    \drawsegment[\leftspace+2.8,6.2*\stepheight,8/5*4,3,,,upper_final,lower_final];
    \drawsegment[\leftspace+2.8+8/5*4,6.2*\stepheight,4,0,,,upper_final,lower_final];

    \draw (\leftspace+2.8,6.2*\stepheight) -- ++ (2/5*4*\leftspace, 0)
        node[midway, above = 0.5] (1) {$\buf_{1}\&\buf'_{1}$}
        node[midway, below = 0.1] (1) {$m'$ cells};
    \draw (\leftspace+2.8+2/5*4,6.2*\stepheight) -- ++ (2/5*4*\leftspace, 0)
        node[midway, above = 0.5] (1) {$\cdots$}
        node[midway, below = 0.3] (1) {$\cdots$};

    \draw (\leftspace+2.8+8/5*4, 6.2*\stepheight) -- ++ (4*\leftspace, 0)
        node[midway, above = 0.5] (1) {$P$};
    
\end{tikzpicture}}
    \caption{An illustration of adaptive allocation. The intervals indicate the buffers, and the shaded regions represent the cells that are filled with items. (i) At the start of phase $1$, we allocate one buffer for each segment. (ii) During phase $1$, items are inserted into the buffers, until one buffer is full. The buffers are not necessarily filled from left to right. (iii) At this point, we allocate new buffers within the pool $P$, where the sizes of the new buffers are given by the Chernoff bound. Note that the new buffers have different sizes, because the segments didn't receive the same number of items in phase $1$. (iv) After this adaptive allocation, as far as the new pool is concerned, the two buffers of each segment behave as one single buffer of size $m'$. This is to say that, replacing these two buffers with one big buffer does not change the set of items inserted to the new pool (however the items inserted to these two buffers are handled differently as those inserted to one big buffer).}\label{fig:re_allocate}
\end{figure}

\paragraph{The algorithm $\algAdapt$.} Our improved algorithm $\algAdapt$ has multiple \defn{phases}, where in each phase we will allocate new buffers. In phase $1$, $\algAdapt$ behaves exactly the same as $\algBase$, where it allocates buffers and inserts new items into the buffers. For each buffer, we will recursively apply $\algAdapt$. 

Phase $1$ ends when some buffer becomes full. At this point, $\algAdapt$ uses the history of the input to predict the number of items from each segment to the pool. Suppose that, during phase $1$, we have processed all but $n'$ items of the input. Then there exists some bound $m'= (n'/B)-\Omega((n'/B)^{2/3})$\footnote{The alternative definition of $m'\defeq(n'/B)-\Omega(\sqrt{n'/B}\cdot \poly\log n)$ would work just as well. We use the current form because it simplifies the computation.} such that, with high probability, for each segment, at least $m'$ of the future $n'$ items belong to this segment. For a segment $i$, if there currently are $c_i$ empty cells in the buffer $\buf_i$, then we can conclude that, the pool receives at least $m'-c_i$ items from segment $i$ (the use of $c_i$ is where the adaptivity comes in). We then allocate a \emph{new} buffer to segment $i$, which occupies $m'-c_i$ consecutive cells in the pool\footnote{In the algorithm, we will choose the parameter $m'$ carefully, so that $m'> c_i$ holds with high probability.}. After this, the pool is redefined to be the unallocated part of the old pool. We then start phase $2$.

A crucial property of this adaptive allocation process is that, at the start of phase $2$, the current state is similar to the very beginning of our algorithm. That is, we currently have two buffers per segment, and these two buffers together are functionally equivalent to one big buffer of size $m'$.

This implies that we can repeat what we did in phase $1$: Resume receiving items, until some new buffer also becomes full. At this point, run another adaptive allocation, then start phase $3$. This goes on until the number of future items drops to $B\cdot \poly\log n$, at which point we can no longer apply the Chernoff bound, because the expected number of future items for each segment is only $\poly\log n$. We then handle what remains of the pool (henceforth called the \defn{final pool}) using a trivial algorithm.

\paragraph{The cost of $\algAdapt$.} Let us fix a segment, and consider the sizes of all the buffers allocated to this segment in $\algAdapt$. We have that the size of the first buffer is $\sim n/B$. Next, the size of the second buffer is $\sim (n/B)^{2/3}$. This is because, since we used Chernoff bound, the actual number of items from a segment only differs from the size of its first buffer by $O((n/B)^{2/3})$. Since the second buffer is also filled with high probability, its size must be no more than this bound. The same logic applies to the other buffers, and we can show that their sizes are $\sim (n/B)^{4/9},\sim (n/B)^{8/27},\dots$, respectively, where the last buffer has size $\poly\log n$. This also implies that there are only $O(\log\log n)$ phases.

Formally, define $T_{\text{adpt}}(n)$ as the expected cost of $\algAdapt$ on input of length $n$. Similar to $\algBase$, if we wish to bound $T_{\text{adpt}}(n)$, then we need to consider cost from three sources.
\begin{enumerate}
    \item \textbf{The (recursive) cost of the buffers:} The cost of recursively implementing each of the buffers is bounded by $T_{\text{adpt}}(n/B)\cdot O(\log\log n)$, where the $\log \log n$ factor comes intuitively from the fact that there are $O(\log \log n)$ phases.
    \item \textbf{The cost of the (final) pool:} We use adaptive allocation to reduce the size of the pool, so the size of the final pool becomes $B\cdot \poly\log n$. Therefore, its cost is at most $B\cdot \poly\log n$.
    \item \textbf{The inter-subarray cost:} Similar to the analysis for $\algBase$, the buffers of the same phase only incur $O(1)$ cost because the buffers are arranged in increasing order of their corresponding segment. Therefore, the inter-subarray cost is of the same order as the number of phases, which is $O(\log\log n)$.
\end{enumerate}

Putting the pieces together, we have
\begin{align*}
    T_{\text{adpt}}(n)\le T_{\text{adpt}}(n/B)\cdot O(\log\log n)+B\cdot \poly\log n+O(\log\log n).
\end{align*}

We set $B$ to be $\Theta(\log n)$, in which case we can show that $T_{\text{adpt}}(n)$ evaluates to $n^{o(1)}$.

\subsection{Merge Segments Using a Synchronization Technique}

Next, building upon the adaptive allocation, we develop a technique for merging segments, that further improves the cost to $\poly\log n$.

\paragraph{Why do we need merging.} The previous algorithm $\algAdapt$ already achieves $n^{o(1)}$ cost bound, but there remain barriers to further improvement.

We first investigate the role of the number of segments $B$ in controlling the cost of $\algAdapt$. By looking at the recurrence formula for $T_{\text{adpt}}$, we can see that: The buffer cost is low when $B$ is large, because then the individual buffers are small, so that the items inserted into the same buffer are closer to each other; The final pool cost is low when $B$ is small, because then we can apply Chernoff bound to a smaller set of items (i.e., when only $B\cdot \poly\log n$ items remain), which allows us to further reduce the size of the pool. However, there is no single value of $B$ that can improve both costs to $\poly\log n$ simultaneously.

To deal with this, we develop a new subroutine for \emph{merging segments}, which allows us to decrease the value of $B$ \emph{during execution}. Intuitively, the goal of merging is to let $B$ be large at the start and small at the end, so that we can simultaneously achieve the best of the two cases.

\paragraph{The (non-recursive) algorithm $\algAdaptMerge$.} In $\algAdapt$, we start with buffer sizes $n/B$, and keep running adaptive allocation, until the buffer sizes drop to $\poly\log n$. The new algorithm $\algAdaptMerge$ is similar to $\algAdapt$, but starts with buffer sizes $\poly\log n$ (equivalently, we start with $B=n/\poly\log n$). At the start of a phase, if the buffer sizes are too small, instead of terminating as in $\algAdapt$, $\algAdaptMerge$ performs a \defn{merge subroutine} to merge segments together, so that the buffer sizes increase back to $\poly\log n$, which allows us to perform more rounds of adaptive allocation. Thus, whereas in $\algAdapt$ the \emph{number} $B$ of segments is invariant, and the \emph{sizes} of the buffers being allocated decreases across phases, the $\algAdaptMerge$ algorithm flips this relationship: now the number $B$ of segments will decrease over time, but the sizes of the buffers being allocated will stay roughly the same (i.e., $\poly\log n$).

In a merge subroutine, we allocate new buffers just as during adaptive allocation. These buffers will be called \defn{dampening buffers} (in contrast, the buffers allocated during an adaptive allocation will be called \defn{regular buffers} from now on). However, due to the implementation of merge subroutines, dampening buffers are very different from regular buffers, in that they do not receive uniformly random items as do regular buffers. 

This creates difficulties when trying to recurse on the dampening buffers. Fortunately, we can dodge this problem for now, as we can show that, unlike $\algAdapt$, the buffers in $\algAdaptMerge$ always have size at most $\poly\log n$. Therefore, we can let $\algAdaptMerge$ be a non-recursive algorithm, and still obtain a cost bound of $\poly\log n$.

\paragraph{Requirement for merging.} 

During a merge subroutine, we merge every $K$ adjacent segments together into a \emph{mega-segment}, so that the number of segments becomes $B/K$ (assume that $K|B$). 

After merging, it should still be possible to perform adaptive allocation. Since adaptive allocation highly relies on the specific distribution of insertions to the pool, what we should guarantee is that, after merging, the distribution of insertions into the pool is similar to what it would have been if the algorithm \emph{started} with $B/K$ segments.

Formally, we define the following \defn{multi-way pool distribution} $\mathcal{P}$ to characterize the pool. The distribution $\mathcal{P}$ is parameterized by integers $n,B,m$, where $n$ is the number of items, $B$ the number of segments, and $m$ the buffer size of each segment. $\mathcal{P}_{n,B,m}$ is a distribution over sequences of length $n-B\cdot m$, generated as follows: 
\begin{itemize}
    \item Sample a uniformly random sequence $x_1,\dots,x_n$, where each items is from $[0,1]$. 
    \item Partition $[0,1]$ into $B$ segments.
    \item Maintain a sequence $y$ of items, initially empty. Consider the items $x_1,\dots,x_n$ one by one. 
    \begin{itemize}
        \item When $x_i$ comes, if we have seen at least $m$ items in $x_1,\dots,x_{i-1}$ that are from the same segment as $x_i$, then $x_i$ is appended to the end of $y$ (i.e., it enters the pool);
        \item Otherwise, $x_i$ is ignored (i.e., it enters a buffer).
    \end{itemize}
    This process continues until the length of $y$ reaches $n-B\cdot m$. $y_1,\dots,y_{n-B\cdot m}$ is the generated sequence.
\end{itemize}

Intuitively, the sequence $x_1,\dots,x_n$ represents the entire input, and the sequence $y_1,\dots,y_{n-B\cdot m}$ represents the items inserted to the pool. The multi-way pool distribution exactly characterizes the distribution of items that are inserted into the pool in $\algBase$. Moreover, in $\algAdapt$, right after an adaptive allocation, the distribution of future insertions into the current pool is also a multi-way pool distribution.

In our algorithm, any merge subroutine starts right after some adaptive allocation step. We thus have that, when merging starts, the current pool receives insertions from some multi-way pool distribution $\mathcal{P}_{n',B,m'}$. Then the goal is that, after merging, this property should be preserved. That is, the current pool should receive insertions from a distribution that is close to the multi-way pool distribution $\mathcal{P}_{n',B/K,m'\cdot K}$, which is similar to the one before merging, but with fewer segments and larger buffers.

\paragraph{A problem: Asynchronism of the sub-segments.} If the two distributions $\mathcal{P}_{n',B,m'}$ and $\mathcal{P}_{n',B/K,m'\cdot K}$ were equal to each other, then merging would be trivial. This is sadly not the case, because the sub-segments of a mega-segment are \emph{asynchronous} in $\mathcal{P}_{n',B,m'}$. 

To see this, we focus on one mega-segment $S$, which contains $S_1,\dots,S_K$ as its sub-segments. When generating $\mathcal{P}_{n',B/K,m'\cdot K}$, we have that, all the sub-segments $S_1,\dots,S_K$ start entering $y$ at the same instant. That is, prior to seeing $m'\cdot K$ items from $S$, no item from $S=S_1\cup\dots\cup S_K$ will enter the pool. After this, any item from $S$ can enter the pool. 

However, in $\mathcal{P}_{n',B,m'}$, items from different sub-segments $S_1,\dots,S_K$ may start entering $y$ at different instants. That is, after seeing a prefix $x_1,\dots,x_i$ of items, it may be that, the next item will enter the pool if it's from $S_1$, but not if it's from $S_2$.

\paragraph{The merge subroutine: Make use of the dampening buffer.} To synchronize the sub-segments, we allocate a new buffer in the pool for each mega-segment, called a \defn{dampening buffer}. Formally, we use a prefix of the current pool to allocate $B/K$ \emph{dampening buffers}, each corresponding to a mega-segment. Each dampening buffer has size $m''$. The \defn{post-dampening pool} is then defined to be the suffix of the old pool that is not used as dampening buffers\footnote{The name post-dampening pool is only used in the overview. We still have the same structure in the technical part, but we will simply call it the ``phase-$j$ pool'' instead.}. When an item $x$ from mega-segment $S$ enters the old pool, if the dampening buffer of $S$ is not full, then $x$ is inserted to it. Only when the corresponding dampening buffer is full will $x$ be inserted into the post-dampening pool.

The basic idea behind dampening buffers is to ``absorb'' insertions until the $B/K$ sub-segments that are being merged have time to ``synchronize'' (meaning that the previous buffers for those segments have all filled up). This means that, even though the entire pool (including both the dampening buffers and the post-dampening pool) receives elements according to a wrong distribution (namely, $\mathcal{P}_{n',B,m'}$), the dampening buffers ensure that the \emph{post-dampening pool} receives elements according to the right distribution (namely, $\mathcal{P}_{n',B/K,(m'\cdot K)+m''}$) with good probability.

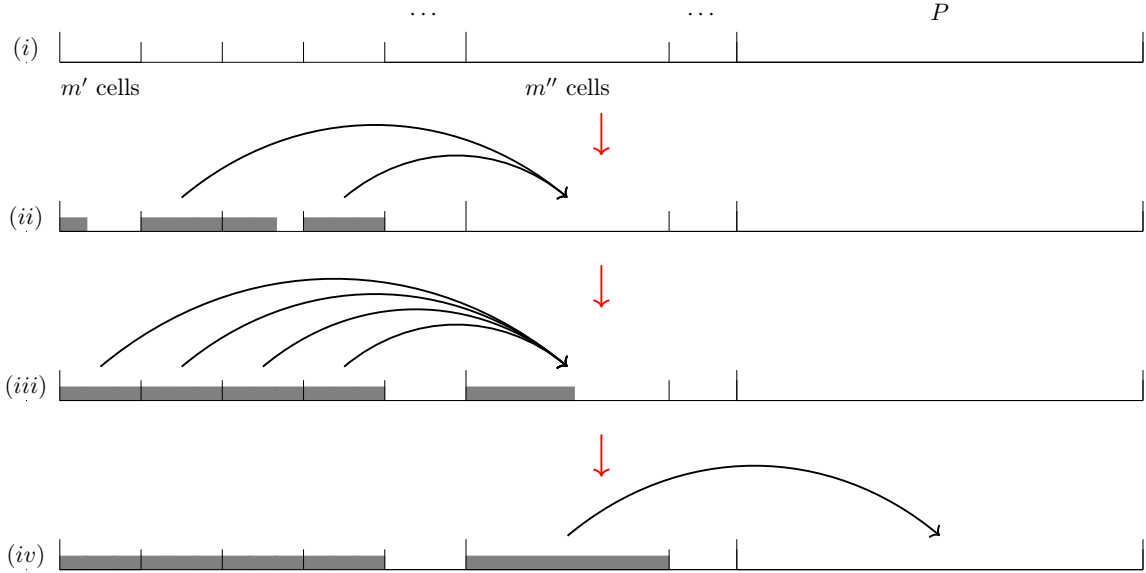
\begin{figure}[ht]
    
    \centering
    \scalebox{0.9}{
    \begin{tikzpicture}
    \def \leftspace {1};
    \def \stepheight {-2.5};
    \def \mytick {0.3};
    \def \drawtick[#1,#2]{\draw ({#1}, {#2}) -- ++ (0,\mytick);} %
    \def \drawlargetick[#1,#2]{\draw ({#1}, {#2}) -- ++ (0,1.5*\mytick);}
    \def \midpoint {\leftspace+8};
    \def \drawsegment[#1,#2,#3,#4,#5,#6,#7,#8]{ %
        \draw ({(#1)}, {(#2)}) -- ({(#1) + (#3)}, {(#2)})
         node[midway, above = 0.1] (#7) {#5}
         node[midway, below = 0.1] (#8) {#6};
        \drawlargetick[(#1),(#2)]
        \drawlargetick[(#1) + (#3),(#2)]
        \foreach \i in {1,...,#4}
        {
            \drawtick[(#1) + \i / (#4 + 1) * (#3),(#2)]
        }
    }

    \draw (\leftspace-0.5, 4*\stepheight) -- ++ (0, 0)
        node[midway, above = -0.1] (1) {$(i)$};

    \def \lenarray {{{1,1,0,1,0},{1,1,1,1,1},{0,1,1,0,1},{1,0,1,1,1}}}

    \drawsegment[\leftspace,4*\stepheight,6,4,,,upper_final,lower_final];

    \drawtick[\leftspace+9,4*\stepheight]
    \drawsegment[\leftspace+10,4*\stepheight,6,0,,,upper_final,lower_final];

    \draw (\leftspace, 4*\stepheight) -- ++ (1.2*\leftspace, 0)
        node[midway, below = 0.1] (1) {$m'$ cells};

    \draw (\leftspace+4.8, 4*\stepheight) -- ++ (1.2*\leftspace, 0)
        node[midway, above = 0.5] (1) {$\cdots$};
        
    \draw (\leftspace+9, 4*\stepheight) -- ++ (\leftspace, 0)
        node[midway, above = 0.5] (1) {$\cdots$};
    \draw (\leftspace+6, 4*\stepheight) -- ++ (3*\leftspace, 0)
        node[midway, below = 0.1] (1) {$m''$ cells};

    \draw (\leftspace+6, 4*\stepheight) -- ++ (10*\leftspace, 0);
    \draw (\leftspace+10, 4*\stepheight) -- ++ (6*\leftspace, 0)
        node[midway, above = 0.5] (1) {$P$};

    \draw [->, color=red, thick] (\midpoint, 4.3 * \stepheight) -- ++ (0, 0.25 * \stepheight);

    \draw (\leftspace-0.5, 5*\stepheight) -- ++ (0, 0)
        node[midway, above = -0.1] (1) {$(ii)$};

    \def \lenarray {{{1,0,0},{1,1,1},{1,1,0},{1,1,1}}}
    
    \foreach \i in {0,...,3}
    {
        \foreach \j in {0,...,2}
        {
        
            \filldraw [fill=gray, draw=gray] (\leftspace+1.2*\i+0.4*\j,5*\stepheight) rectangle ++(0.4*\lenarray[\i][\j],0.2*\lenarray[\i][\j]);
        }
    }

    \drawsegment[\leftspace,5*\stepheight,6,4,,,upper_final,lower_final];

    \drawtick[\leftspace+9,5*\stepheight]
    \drawsegment[\leftspace+10,5*\stepheight,6,0,,,upper_final,lower_final];

    \draw (\leftspace, 5*\stepheight) -- ++ (1.2*\leftspace, 0);

    \draw (\leftspace+4.8, 5*\stepheight) -- ++ (1.2*\leftspace, 0);
        
    \draw (\leftspace+9, 5*\stepheight) -- ++ (\leftspace, 0);
    \draw (\leftspace+6, 5*\stepheight) -- ++ (3*\leftspace, 0);

    \draw (\leftspace+6, 5*\stepheight) -- ++ (10*\leftspace, 0);
    \draw (\leftspace+10, 5*\stepheight) -- ++ (6*\leftspace, 0);

    \draw [->, thick] (\leftspace+1.8, 4.8 * \stepheight) to[out=40,in=140] (\leftspace+7.5, 4.8 * \stepheight);
    
    \draw [->, thick] (\leftspace+4.2, 4.8 * \stepheight) to[out=40,in=140] (\leftspace+7.5, 4.8 * \stepheight);

    \draw [->, color=red, thick] (\midpoint, 5.2 * \stepheight) -- ++ (0, 0.25 * \stepheight);

    \draw (\leftspace-0.5, 6*\stepheight) -- ++ (0, 0)
        node[midway, above = -0.1] (1) {$(iii)$};

    \foreach \i in {0,...,3}
    {
        \foreach \j in {0,...,2}
        {
        
            \filldraw [fill=gray, draw=gray] (\leftspace+1.2*\i+0.4*\j,6*\stepheight) rectangle ++(0.4,0.2);
        }
    }

    \filldraw [fill=gray, draw=gray] (\leftspace+6,6*\stepheight) rectangle ++(1.6,0.2);

    \drawsegment[\leftspace,6*\stepheight,6,4,,,upper_final,lower_final];

    \drawtick[\leftspace+9,6*\stepheight];
    \drawsegment[\leftspace+10,6*\stepheight,6,0,,,upper_final,lower_final];

    \draw (\leftspace, 6*\stepheight) -- ++ (1.2*\leftspace, 0);

    \draw (\leftspace+4.8, 6*\stepheight) -- ++ (1.2*\leftspace, 0);
        
    \draw (\leftspace+9, 6*\stepheight) -- ++ (\leftspace, 0);
    \draw (\leftspace+6, 6*\stepheight) -- ++ (3*\leftspace, 0);

    \draw (\leftspace+6, 6*\stepheight) -- ++ (10*\leftspace, 0);
    \draw (\leftspace+10, 6*\stepheight) -- ++ (6*\leftspace, 0);

    \draw [->, thick] (\leftspace+1.8, 5.8 * \stepheight) to[out=40,in=140] (\leftspace+7.5, 5.8 * \stepheight);
    
    \draw [->, thick] (\leftspace+4.2, 5.8 * \stepheight) to[out=40,in=140] (\leftspace+7.5, 5.8 * \stepheight);
    \draw [->, thick] (\leftspace+0.6, 5.8 * \stepheight) to[out=40,in=140] (\leftspace+7.5, 5.8 * \stepheight);
    \draw [->, thick] (\leftspace+3, 5.8 * \stepheight) to[out=40,in=140] (\leftspace+7.5, 5.8 * \stepheight);

    \draw [->, color=red, thick] (\midpoint, 6.2 * \stepheight) -- ++ (0, 0.25 * \stepheight);

    \draw (\leftspace-0.5, 7*\stepheight) -- ++ (0, 0)
        node[midway, above = -0.1] (1) {$(iv)$};

    \foreach \i in {0,...,3}
    {
        \foreach \j in {0,...,2}
        {
        
            \filldraw [fill=gray, draw=gray] (\leftspace+1.2*\i+0.4*\j,7*\stepheight) rectangle ++(0.4,0.2);
        }
    }

    \filldraw [fill=gray, draw=gray] (\leftspace+6,7*\stepheight) rectangle ++(3,0.2);

    \drawsegment[\leftspace,7*\stepheight,6,4,,,upper_final,lower_final];

    \drawtick[\leftspace+9,7*\stepheight];
    \drawsegment[\leftspace+10,7*\stepheight,6,0,,,upper_final,lower_final];

    \draw (\leftspace, 7*\stepheight) -- ++ (1.2*\leftspace, 0);

    \draw (\leftspace+4.8, 7*\stepheight) -- ++ (1.2*\leftspace, 0);
        
    \draw (\leftspace+9, 7*\stepheight) -- ++ (\leftspace, 0);
    \draw (\leftspace+6, 7*\stepheight) -- ++ (3*\leftspace, 0);

    \draw (\leftspace+6, 7*\stepheight) -- ++ (10*\leftspace, 0);
    \draw (\leftspace+10, 7*\stepheight) -- ++ (6*\leftspace, 0);

    \draw [->, thick] (\leftspace+7.5, 6.8 * \stepheight) to[out=40,in=140] (\leftspace+13, 6.8 * \stepheight);
    
\end{tikzpicture}}
    \caption{An illustration of how dampening buffers work. In this example, $K=4$, and we only illustrate the behavior of the first mega-segment. Similar to \cref{fig:re_allocate}, the shaded regions represent the filled cells. The arrows from one buffer to another indicates that, since the previous buffer is full, new items that should enter it are instead forwarded to the other buffer. (i) Before merging, we have a set of buffers, where the buffers of each segment has $m'$ empty cells in total. For each mega-segment, we allocate one dampening buffer of size $m''$. $P$ is the post-dampening pool.  (ii,iii) New items first enter the older buffers. Only when the buffers of some sub-segment are all full will new items from that sub-segment start entering the dampening buffer. (iv) Finally, new items from the mega-segment only start entering the post-dampening pool when the dampening buffer is full. Ideally, when the dampening buffer is full, the older buffers are also full.}\label{fig:dampening}
\end{figure}

To prove this, we note that if the size of the dampening buffers $m''$ is sufficiently large, then the distribution of insertions to the \emph{post-dampening pool} is close to $\mathcal{P}_{n',B/K,(m'\cdot K)+m''}$. In other words, for a mega-segment, the buffers from its $K$ sub-segments and the dampening buffer together behave like one big buffer of size $(m'\cdot K)+m''$.

Indeed, if $m''$ is sufficiently large, then a Chernoff bound tells us that, when the dampening buffer of $S$ becomes full, with high probability, we would have received at least $m'$ items from each sub-segment $S_i$, which means that all the previous buffers must already be full. In other words, until we have seen $(m'\cdot K)+m''$ items from $S$, the dampening buffer will not be full. Thus, the first $(m'\cdot K)+m''$ items from $S$ will be inserted into the buffers, which means that the insertions to the \emph{post-dampening buffer} does (approximately) follow the desired distribution.

\subsection{Recursion for the Dampening Buffers} 
As noted above, the reason that we cannot na\"ively transform the above algorithm into a recursive algorithm is becasue the dampening buffers do not receive i.i.d.~unformly random insertions. 

Notice, however, that even though the insertions into a \emph{dampening buffer} are not uniformly random, they are still highly structured. In fact, the insertions that a dampening buffer receives follow a distribution that is close to being a multi-way pool distribution---the same type of distribution that we normally associate with an entire pool! 

This suggests a possible way to add recursion back into our algorithm: rather than defining recursive subproblems to assume i.i.d.~uniformly random elements, we can define a richer recursive structure in which subproblems receive elements according to a certain type of multi-way pool distribution. 

As we see in Section \ref{sec:final}, there are several technical challenges to making such a recursion work. In order to meaningfully make use of an input that is drawn according to a multi-way pool distribution, a recursive subproblem must be given not just its own input elements, but information about the earlier elements that are not inserted into the pool (and were instead inserted into older buffers). Nonetheless, with sufficient care, we show that it is still possible to transform $\algAdaptMerge$ into a recursive algorithm. 

Combining this high-level approach with a careful analysis, we are able to design a recursive algorithm whose expected cost is bounded by $\log n\cdot 2^{O(\log^* n)}$. The $2^{O(\log^* n)}$ factor comes from the recursive structure: Intuitively, if $h(n)$ denotes the expected cost on a subproblem of size $n$, then we have $h(n) = O(\log n / \log \log n) \cdot h(\poly\log n)$, where the $\log n / \log \log n$ term comes from the number of phases, and the $h(\poly\log n)$ comes from the fact that we recurse on subproblems of size $\poly\log n$ (because the buffer sizes are $\poly\log n$). This solves to $h(n) = \log n\cdot 2^{O(\log^* n)}$.

We are also able to prove a nearly-matching lower bound of $\Omega(\log n)$ in \cref{sec:lb}, meaning that our final bound is necessarily optimal up to a $2^{O(\log^* n)}$ factor.

\section{The First Technique: Adaptive Allocation}
\label{sec:warm_up}

In this section, we present an algorithm $\algAdapt$ that achieves $n^{o(1)}$ competitive ratio. As discussed in the technical overview, $\algAdapt$ works by adaptively allocating new buffers in order to reducing the size of the pool all the way down to $B\cdot \poly\log n$. In the rest of this section, we formalize this algorithm, along with its analysis.

\subsection{Notations}

Our algorithm partitions the timeline into multiple \defn{phases}. In each phase, we make predictions about the future, and use it to allocate new buffers. Let $B$ be the number of segments, and let $n_j$ be the number of future items at the start of phase $j$ (in particular, $n_1=n$). At phase $j$, we use
\begin{align*}
    m_j\defeq \biggr\lceil\bk*{\frac {n_j}B}-\frac 12\cdot\bk*{\frac {n_j}B}^{2/3}\biggr\rceil
\end{align*}
as a lower bound on the number of future items from each segment. Note, in particular, that so long as $n_j \gg B \log^3 n$, this lower bound holds with high probability in $n$. In the future, for notational simplicity, we will write $m_j$ as $(n_j/B)-d(n_j/B)$, where $d(x)\defeq x-\lceil x-\frac 12\cdot x^{2/3}\rceil$. Note that $d(x)=\Theta(x^{2/3})$, and that $m_j$ is always an integer.

At the start of phase $j$, we will allocate a new buffer for each segment, where the buffer for segment $i$ is denoted as $\buf_{i,j}$. At any point in time, for a buffer $\buf_{i,j}$, let its \defn{remaining capacity} be the number of remaining empty cells in it. The \defn{remaining capacity of a segment} $i$ is defined to be the sum of remaining capacities of all the buffers for segment $i$, and is denoted as $c_i$.

\subsection{The Algorithm $\algAdapt$}

In phase $1$, $\algAdapt$ behaves just as $\algBase$. Partition the $n$ cells into $B$ buffers $\buf_{1,1},\dots,\buf_{B,1}$ and one pool $P$. Each buffer $\buf_{i,1}$ has size $m_1= (n_1/B)-d(n_1/B)$ where $n_1=n$. The items are processed in the same way as before, where within each buffer, we recursively call $\algAdapt$ to handle the insertions. We end phase $1$ when some buffer becomes full, at which point we perform an \defn{adaptive allocation}, and start phase $2$.

\paragraph{Adaptive allocation.} When phase $1$ ends, we have processed all but $n_2$ items. Using a Chernoff bound, we have that with high probability, for each segment, at least $m_2=\bk*{n_2/B} - d\bk*{n_2/B}$ of the $n_2$ future items belong to this segment. Recall that $c_i$ is the current number of remaining empty cells in $\buf_{i,1}$. Thus the Chernoff bound implies that, with high probability, at least $m_2-c_i$ items inserted to the pool will be from segment $i$.

Given this knowledge, we allocate $m_2-c_i$ cells of the pool to segment $i$. Formally, we take a prefix of the pool $P$, and partition it into $B$ buffers $\buf_{1,2},\dots,\buf_{B,2}$, where each buffer $\buf_{i,2}$ contains $m_2-c_i$ cells. The pool $P$ is then redefined to be the unallocated part of the old pool. This process is illustrated in \cref{fig:re_allocate}. The case where $m_2-c_i\le 0$ for some $i$ is unlikely, and is handled by the failure mode described below.

After this adaptive allocation, we now have that the remaining capacity $c_i$ of each segment is equal to $m_2$. Although there are two buffers per segment, they are functionally equivalent to one single buffer of remaining capacity $c_i$. Thus we arrive at a state similar to the beginning of phase $1$, and we start phase $2$ from here.

\paragraph{Repeated allocation.} In phase $2$, upon receiving an item that belongs to some segment, we will insert it in the \emph{oldest} non-full buffer of that segment (that is, we only insert it in $\buf_{i,2}$ when $\buf_{i,1}$ is already full). Again, within each new buffer, $\algAdapt$ is called recursively to process the insertions. Phase $2$ ends when the remaining capacity of some segment drops to $0$, at which point we run another adaptive allocation, and start phase $3$.

This process continues until the size of the pool becomes very small. If, \emph{after} the adaptive allocation of phase $j$, the size of the current pool\footnote{This is slightly different from what we described in the technical overview, where we used the condition that the \emph{number of future items} is small. It turns out that these values are of the same order, and that it is easier for the analysis to check the size of the pool.} is smaller than $B\cdot \log ^5n$, then we let $j$ be the last phase, and refer to the current pool as the \defn{final pool}. In this last phase, when the remaining capacity of some segment becomes $0$, any future item from this segment will be inserted to the remaining first empty cell in the final pool.

Note that, for any phase $j$, prior to adaptive allocation, the size of the pool is at least $B\cdot \log^5 n$. This implies that, we must have $n_j\ge B\cdot \log^5 n$ for any phase $j$.

\paragraph{Failure mode.} If the input does not satisfy certain concentration properties, then our algorithm would behave poorly. In those cases, the algorithm will just give up and enter the failure mode, in which we completely forget about the buffer structure, and any new item is simply inserted into the first empty cell in the array. Formally, we enter the failure mode immediately after encountering any of the following events:

\begin{enumerate}
    \item Between consecutive phases, the number of future items does not decrease by as much as one would expect, namely, $(n_{j+1}/B)>(n_j/B)^{2/3}$.
    \item During adaptive allocation, the size of some new buffer $m_j-c_i$ is nonpositive.
    \item When we process some item $x$ from segment $i$, there are no cells for it. That is, both the buffers of segment $i$ and the final pool are full. Note that this can only happen in the last phase, because in the previous phases, the remaining capacities of segments are always positive (indeed, when one of the remaining capacities hits zero, that's when we start a new phase with new buffers).
\end{enumerate}

If the algorithm never enters the failure mode, we say that it succeeds. Note that, even if the algorithm's recursive subproblems (i.e., the recursions on buffers) fail, so long as the current level of recursion \emph{does not fail}, we still consider the algorithm to have succeeded. 

\subsection{The analysis}

\paragraph{Proof of correctness.} To prove that the algorithm succeeds with high probability, it suffices to show the following claim, which implies that during each phase, the arriving items exhibit nice concentration properties.

\begin{claim}
    \label{clm:warm_up_phase_concentration}
    Let $n,B$ be integers where $1<B\le n/\log^5 n$, and let $n'$ be an integer satisfying $B\cdot \log^5n\le n'\le n$. Let $m'=(n'/B)-d(n'/B)$. Consider the following random process: Start with $n'$ balls and $B$ bins, where the bins are initially empty. The balls are placed one after another into random bins. For this random process, with probability $1-O(1/n^{11})$, the following properties hold:
    \begin{itemize}
        \item After all the balls are placed, each bin has load at least $m'$.
        \item Consider the first moment that the load of some bin reaches $m'$. At this moment, the load of any other bin is at least $m'-(1/5)\cdot d(n'/B)$.
    \end{itemize}
\end{claim}

\begin{proof}
    WLOG assume that $n$ is sufficiently large. By a Chernoff bound, the following holds with probability $1-O(1/n^{11})$ (recall that $d(x)=\Omega(x^{2/3})$): For any $i$ such that $B\cdot \log^4n\le i\le n'$, after the first $i$ items are placed into bins, the load of each bin is in the range
    \begin{align*}
        \frac iB\pm \frac 1{20}\cdot d\bk*{\frac iB}.
    \end{align*}
    Now we prove that the claimed properties hold (deterministically) so long as the above concentration bound holds. For the first bullet, it suffices to apply the concentration bound at time $n'$.
    
    For the second bullet, let $i^*$ be the first moment that the load of some bin reaches $m'$. Using the concentration property, at time $\lceil B\cdot \log^4n\rceil$, each bin has load $O(\log^4n)\ll m'$. Therefore, $i^*$ must be at least $B\cdot \log^4n$, which means that we can apply the concentration bound at time $i^*$.

    Since one of the bins has load $m'$, using the upper bound side of the concentration bound, we have that at time $i^*$,
    \begin{align*}
        \frac {i^*}B+ \frac 1{20}\cdot d\bk*{\frac {i^*}B}\ge m'.
    \end{align*}
    Therefore, by applying the lower bound side of the concentration bound, any other bin has load at least
    \begin{align*}
        &\frac {i^*}B-\frac 1{20}\cdot d\bk*{\frac {i^*}B} \\
        \ge{}&m'-\frac 1{10}\cdot d\bk*{\frac {i^*}B} \\
        \ge{}&m'-\frac 15\cdot d\bk*{\frac {n'}B}.\tag{for sufficiently large $n$}
    \end{align*}
\end{proof}

Given \cref{clm:warm_up_phase_concentration}, the proof of correctness proceeds by applying the claim to each phase.

\begin{lemma}
    \label{lem:warm_up_regularity}
    When $1<B\le n/\log^5n$, the probability that $\algAdapt$ fails is $O(1/n^{10})$. 
\end{lemma}

\begin{proof}
    Suppose that we are currently at the start of phase $j$, and that the algorithm does not fail up to this point. We show that with probability $1-O(1/n^{11})$, $\algAdapt$ will not fail until after the buffers of phase $j+1$ are allocated, or if $j$ is the last phase, then we show that $\algAdapt$ will successfully terminate. Then, since there can be at most $n$ phases, a union bound over all phases proves the lemma.

    Since we always have that $n_j\ge B\cdot \log^5n$, we can apply \cref{clm:warm_up_phase_concentration} on the items that arrive during and after phase $j$. In the following, we use \cref{clm:warm_up_phase_concentration} to show that each failure is unlikely.

    \begin{enumerate}
        \item For the first type of failure (when $j$ is not the last phase): \cref{clm:warm_up_phase_concentration} implies that, by the time that phase $j$ ends, the remaining capacity of any buffer is at most $(1/5)\cdot d(n_j/B)$. When phase $j$ ends, the number of future items $n_{j+1}$ is precisely the current number of empty cells in the array, which is the sum of the current pool size and the remaining capacities of the buffers. The pool size during phase $j$ is the total number of future items $n_j$ minus the remaining capacities at the start of phase $j$, which is $B\cdot m_j$. Therefore, $n_{j+1}$ can be bounded as
        \begin{align*}
            n_{j+1}\le{}&(n_j-B\cdot m_j)+B\cdot \frac 15 \cdot d\bk*{\frac{n_j}B} \\
            \le {}&\frac 65\cdot B\cdot d\bk*{\frac{n_j}B}\tag{definition of $m_j$} \\
            \le {}&\frac 12\cdot \frac 65\cdot B\cdot \bk*{\frac{n_j}B}^{2/3}\le B\cdot \bk*{\frac {n_j}B}^{2/3}.\tag{definition of $d(x)$}
        \end{align*}
        Therefore, the algorithm avoids the first type of failure. 
        \item For the second type of failure (when $j$ is not the last phase): We have established that, with high probability, when phase $j$ ends, the remaining capacity of each buffer is at most $(1/5)\cdot d(n_j/B)$. It then remains to show that $m_{j+1}> (1/5)\cdot d(n_j/B)$. By definition, we have that $m_{j+1}=(n_{j+1}/B)-d(n_{j+1}/B)$. We also have that $n_{j+1}\ge B\cdot \log^5n$. This implies that for sufficiently large $n$, $m_{j+1}\ge (1/2)\cdot (n_{j+1}/B)$. Moreover, since $n_{j+1}$ is exactly the number of empty cells in the array, which is at least the size of the pool during phase $j$, we have the following lower bound:
        \begin{align*}
            n_{j+1}\ge{}&n_j-B\cdot m_j \\
            = {}& B\cdot d\bk*{\frac{n_j}B}.\tag{definition of $m_j$}
        \end{align*}
        Combining the two bounds, we conclude that $m_{j+1}> (1/5)\cdot d(n_j/B)$ holds with high probability.
        \item For the third type of failure (when $j$ \emph{is} the last phase): For each segment $i$, let $v_i$ be the number of items from $i$ that arrive during the last phase $j$. Since each segment start with remaining capacity $c_i\equiv m_j$, the number of items that overflow from their corresponding buffers is $\sum_{i=1}^{B} \max(v_i-m_j,0)$. If the third type of failure happens, then we must have that the number of items inserted to the pool exceeds the size of the pool, which is $\sum_{i=1}^{B} \max(v_i-m_j,0)> n_j-B\cdot m_j$. Since $\sum_{i=1}^{B}v_i=n_j$, this only happens when $v_i<m_j$ for some $i$. That is, for some segment $i$, the number of items that arrive during the last phase and are from segment $i$ is smaller than $m_j$. According to \cref{clm:warm_up_phase_concentration}, this happens with negligible probability. \qedhere
    \end{enumerate}

\end{proof}

\paragraph{Proof of efficiency.} As discussed in the technical overview, we have the following properties for $\algAdapt$:

\begin{lemma}
    \label{lem:warm_up_decrement}
    Conditioned on $\algAdapt$ succeeding, we have that:
    \begin{itemize}
        \item Any buffer $\buf_{i,j}$ will be filled by uniformly random items from segment $i$.
        \item The size of buffer $\buf_{i,j}$ is at most $(n/B)^{(2/3)^{j-1}}$. Consequently, the number of phases is $O(\log\log n)$.
    \end{itemize}
\end{lemma}

\begin{proof}
    For the first bullet: If for some buffer $\buf_{i,j}$, any item inserted to it is not from segment $i$, then this item would trigger the third type of failure. Thus, each buffer only receives items from its own segment. 
    
    To see that these items are uniformly random, we note that whether the algorithm succeeds is independent of the value of items modulo their segments. That is, we can think of the generation of the input items as first sampling $y_i\in [B]$, the segment of $x_i$, and $z_i\in [0,1/B]$, the value of $x_i$ within its segment, then setting $x_i=(y_i-1)/B+z_i$. When deciding whether to insert an item into a buffer, we only care about the segments that each $x_i$ belong to, i.e., the $y_i$'s. Therefore, for the items inserted to a buffer, their $z_i$'s are uniformly random.

    For the second bullet, we have that the size of $\buf_{i,j}$ is bounded by $m_j\le n_j/B$. Since $n_1= n$ and $(n_j/B)\le (n_{j-1}/B)^{2/3}$ (by the definition of the failure mode), it follows that $m_j\le (n/B)^{(2/3)^{j-1}}$.
\end{proof}

Given \cref{lem:warm_up_regularity} and \cref{lem:warm_up_decrement}, we can analyze the competitive ratio of $\algAdapt$:

\begin{theorem}
    When $B$ is set to be $\Theta(\log n)$, the expected competitive ratio of $\algAdapt$ is $n^{o(1)}$.
\end{theorem}

\begin{proof}
    Denote the expected cost of $\algAdapt$ on $n$ cells by $T_{\text{adpt}}(n)$.
    
    Fix a constant $\epsilon>0$. We show that, there exists a constant $C$ such that $T_{\text{adpt}}(n)\le C\cdot n^{\epsilon}$ for any $n$. We only have to prove this for sufficiently large $n$, since we can let $C$ be arbitrarily large. In the remainder of the proof, the big-O notations do not hide $\epsilon$ and $C$.

    The proof is by induction. Suppose that for a sufficiently large $n$, we have $T_{\text{adpt}}(n')\le C\cdot (n')^{\epsilon}$ for any $n'<n$. We now show that the same holds for $n$.

    The cost of the algorithm comes from the following parts:
    \begin{enumerate}
        \item In the case that $\algAdapt$ fails, its cost is at most $n$. The expected cost from this part is $o(1)$.
        \item \textbf{The cost of the buffers:} Within each buffer, we recursively apply $\algAdapt$. The expected cost of each buffer is 
        \begin{align*}
            T_{\text{adpt}}(\text{buffer length})/B.
        \end{align*}
        Using \cref{lem:warm_up_decrement} to bound the buffer lengths, we have that, the total cost of buffers is at most
        \begin{align*}
            &\sum_{j=0}^{O(\log\log n)}\sum_{i=1}^{B}T_{\text{adpt}}(|\buf_{i,j}|)/B \\
            \le{}& C\cdot \bk*{\bk*{\frac nB}^{\epsilon}+\bk*{\frac nB}^{(2/3)\cdot\epsilon}+\cdots} \\
            ={}&C\cdot O\bk*{\log\log n\cdot\bk*{\frac n{\log n}}^{\epsilon}}\le (1/2)C\cdot n^\epsilon.
        \end{align*}
        \item \textbf{The cost of the final pool:} Within the final pool, the cost is at most its size, which is $B\cdot \log^8 n=\poly\log n$.
        \item \textbf{The inter-subarray cost:} The total difference between adjacent buffers (and between the last buffer and the pool) is $O(\log\log n)$. This is because, the total difference between subarrays of the same phase is $O(1)$ (they are arranged in increasing order of their corresponding segment), and there are $O(\log\log n)$ phases.
    \end{enumerate}
    Combining these parts, we have that $T_{\text{adpt}}(n)\le C\cdot n^{\epsilon}$ for any $n$.
\end{proof}

\section{The Second Technique: Merge Segments Using a Synchronization Technique}
\label{sec:warm_up_2}

In this section, we design a non-recursive algorithm $\algAdaptMerge$, that achieves competitive ratio $\poly\log n$. An interesting feature of $\algAdaptMerge$ is that, because of the non-recursiveness, it achieves a $\poly\log n$ competitive ratio not just in expectation, but with \emph{high-probability}.

\highProb*

We will see in Section \ref{sec:final} how to bring this competitive ratio down to almost $\log n$, by adding an intricate recursive structure to the current algorithm.

As mentioned in the technical overview, $\algAdaptMerge$ works by adding a merge subroutine to the algorithm, which allows us to keep the sizes of all the buffers (as well as the final pool) to be $\poly\log n$.

\subsection{Notations}

The new algorithm $\algAdaptMerge$ builds upon (the non-recursive version of) $\algAdapt$, with the additional feature that the number of segments changes from phase to phase. However, within each phase, the number of segments is fixed.

Let $B_j$ denote the number of segments during phase $j$. Since the segments change over time, we will sometimes refer to a segment as a \defn{phase-$j$} segment, which means that it is of the form $[(i-1)/B_j,i/B_j]$ for some integer $i$. The adjective ``phase-$j$'' is omitted if it is clear from context. We say that a segment is a \defn{current segment}, if we are currently in some phase $j$ and the segment is a phase-$j$ segment. Let $K=\lceil \log^8 n\rceil$ be a parameter. We will guarantee that $B_j$ is non-increasing, and that $B_j$ is always a power of $K$, so that any phase-$(j-1)$ segment is fully contained in one of the phase-$j$ segments. The initial number of segments $B_1$ is set to be the largest power of $K$ such that $B_1\cdot K\le n$.

The definitions of $n_j,m_j$ remain the same: $n_j$ is the number of future items at the start of phase $j$, and $m_j=(n_j/B_j)-d(n_j/B_j)$ is a lower bound used in the algorithm. We also define another bound $m_j^{\text{old}}=(n_j/B_{j-1})-d(n_j/B_{j-1})$. In the algorithm, both bounds $m_j,m_j^{\text{old}}$ will be useful when we perform a merge subroutine at the start of phase $j$, where number of segments changes from $B_{j-1}$ to $B_{j}$.

For a buffer, its \defn{initial segment} is defined to be the segment that it was allocated to. We say that a buffer \defn{accepts} an item, if and only if the item is from its initial segment. At any point in time, we say that it \defn{belongs} to the (unique) current segment that contains its initial segment.

At any point in time, $c_i$ denotes the \defn{remaining capacity of current segment $i$}, and is defined to be the sum of the remaining capacities of all the buffers that currently \emph{belong} to $i$.

\subsection{The Algorithm $\algAdaptMerge$}

As before, Phase $1$ of $\algAdaptMerge$ begins by allocating $B_1$ buffers, each of size $m_1=(n_1/B_1)-d(n_1/B_1)$. 

The condition for ending a phase is the same as before: That is, when the remaining capacity of some current segment becomes $0$. When phase $j-1$ ends, we first run an adaptive allocation for every phase-$(j-1)$ segment as in $\algAdapt$, so that each phase-$(j-1)$ segment now has remaining capacity $m_j^{\text{old}}=(n_j/B_{j-1})-d(n_j/B_{j-1})$. After this, one of the following happens:
\begin{itemize}
    \item If the size of the pool is at least $B_{j-1}\cdot K$, then we do not merge segments, in which case $B_j=B_{j-1}$ (i.e., the phase-$j$ segments coincide with the phase-$(j-1)$ segments). We then directly start phase $j$.
    \item Otherwise, if the pool is smaller than $B_{j-1}\cdot K$ and $B_{j-1}>K$, we run the merge subroutine described below, so that the number of segments $B_j$ becomes $B_{j-1}/K$ (note that we must have $K|B_{j-1}$), and each phase-$j$ segment now has remaining capacity $m_j$. We then start phase $j$.
    \item If the pool is smaller than $B_{j-1}\cdot K$ and $B_{j-1}\le K$, then we let the current phase $j$ be the last phase, and start phase $j$, during which any overflown item will be inserted into the \emph{final pool}.
\end{itemize}

\paragraph{The merge subroutine.} Suppose that we are right after the adaptive allocation of some phase $j$. That is, we have $n_j$ future items and $B_{j-1}$ phase-$(j-1)$ segments, where the $i$-th segment has remaining capacity $c_i\equiv m_j^{\text{old}}$. After merging, we will have $B_{j-1}/K$ phase-$j$ segments, and the remaining capacity of each phase-$j$ segment will be $m_j$.

Currently, for each phase-$j$ segment, its remaining capacity is $Km_j^{\text{old}}$ (from the buffers that used to belong to the phase-$(j-1)$ segments). To satisfy the capacity requirement, for each phase-$j$ segment, we allocate another buffer of size $m_j-Km_j^{\text{old}}$, which is called a \defn{dampening buffer}. To distinguish between two types of buffers, the ones allocated during an adaptive allocation will be called \defn{regular buffers}. The case where $m_j-Km_j^{\text{old}}\le 0$ is handled by the failure mode (actually, this will not happen for large $n$).

Similar to the regular buffers, the dampening buffers of the same phase are also arranged in increasing order of their corresponding segments, in order to minimize the inter-subarray cost. We use $\buf_{i,j}$ to denote the \emph{regular} buffer that is allocated at the start of phase $j$ for phase-$(j-1)$ segment $i$, and use $\dbuf_{i,j}$ to denote the \emph{dampening} buffer (if it exists) that is allocated at the start of phase $j$ for phase-$j$ segment $i$.

\paragraph{Processing items.} Upon arrival, each item will be inserted in the oldest non-full buffer that \emph{accepts} the item. Note that we do not distinguish between regular and dampening buffers here. We assume that no item lies on the endpoints of any segment; this happens with probability $1$. If no non-full buffer accepts the new item and the final pool is available (i.e., we are in the last phase), then the item is inserted into the final pool. 

Finally, within the buffer (or the final pool), instead of running a recursive algorithm to place the item like in $\algAdapt$, the new item is simply inserted into the first empty cell.

We remark that this method of inserting items is in accordance with our description in the overview; Fix a phase $j$ in which we performed merging. Under our algorithm, when an item arrives, it will only be inserted into the phase-$j$ pool (which corresponds to the post-dampening pool in the overview) if its corresponding phase-$j$ dampening buffer is full.

\paragraph{Failure mode.} The failure mode here is similar to $\algAdapt$. For completeness, we repeat the criteria for entering the failure mode.

\begin{enumerate}
    \item During some phase $j$, the number of future items does not decrease fast enough: $(n_{j+1}/B_j)>(n_j/B_j)^{2/3}$.
    \item Some new buffer has nonpositive size:
    \begin{enumerate}
        \item During any adaptive allocation, the size of some new regular buffer (which is $m_j^{\text{old}}-c_i$ for the $i$-th segment) is nonpositive.
        \item During any merge subroutine, the size of some new dampening buffer (which is $m_j- Km_j^{\text{old}}$) is nonpositive.
    \end{enumerate}
    \item When we process some item $x$, there is no cell for it. That is, no non-full buffer accepts $x$, and the final pool is not available (either because we are not in the last phase yet, or because the final pool is already full).
\end{enumerate}

\begin{figure}[ht]
    \centering
    \scalebox{0.9}{
    \begin{tikzpicture}
    \def \leftspace {1};
    \def \stepheight {-2.5};
    \def \mytick {0.3};
    \def \drawtick[#1,#2]{\draw ({#1}, {#2}) -- ++ (0,\mytick);} %
    \def \drawlargetick[#1,#2]{\draw ({#1}, {#2}) -- ++ (0,1.5*\mytick);}
    \def \midpoint {\leftspace+8};
    \def \drawsegment[#1,#2,#3,#4,#5,#6,#7,#8]{ %
        \draw ({(#1)}, {(#2)}) -- ({(#1) + (#3)}, {(#2)})
         node[midway, above = 0.1] (#7) {#5}
         node[midway, below = 0.1] (#8) {#6};
        \drawlargetick[(#1),(#2)]
        \drawlargetick[(#1) + (#3),(#2)]
        \foreach \i in {1,...,#4}
        {
            \drawtick[(#1) + \i / (#4 + 1) * (#3),(#2)]
        }
    }
    
    \drawsegment[\leftspace,3*\stepheight,8,3,,,upper_final,lower_final];
    \drawsegment[\leftspace+8,3*\stepheight,8,0,,,upper_final,lower_final];
    
    \draw (\leftspace-0.5, 3*\stepheight) -- ++ (0, 0)
        node[midway, above = -0.1] (1) {$(i)$};

    \draw (\leftspace, 3*\stepheight) -- ++ (2*\leftspace, 0)
        node[midway, above = 0.57] (1) {segment $1$}
        node[midway, below = 0.1] (1) {$m_j^{\text{old}}$};

    \draw (\leftspace+2, 3*\stepheight) -- ++ (2*\leftspace, 0)
        node[midway, above = 0.57] (1) {$\dots$}
        node[midway, below = 0.4] (1) {$\dots$};

    \draw (\leftspace+8, 3*\stepheight) -- ++ (8*\leftspace, 0)
        node[midway, above = 0.5] (1) {$P$};

    \draw [->, color=red, thick] (\midpoint, 3.2 * \stepheight) -- ++ (0, 0.25 * \stepheight);

    \draw (\leftspace-0.5, 4*\stepheight) -- ++ (0, 0)
        node[midway, above = -0.1] (1) {$(ii)$};
    \drawsegment[\leftspace,4*\stepheight,8,3,,,upper_final,lower_final];
    \drawsegment[\leftspace+8,4*\stepheight,4,1,,,upper_final,lower_final];
    \drawsegment[\leftspace+12,4*\stepheight,4,0,,,upper_final,lower_final];

    \draw (\leftspace, 4*\stepheight) -- ++ (2*\leftspace, 0)
        node[midway, above = 0.57] (1) {segment $1$}
        node[midway, below = 0.1] (1) {$m_j^{\text{old}}$};

    \draw (\leftspace+2, 4*\stepheight) -- ++ (2*\leftspace, 0)
        node[midway, above = 0.57] (1) {$\dots$}
        node[midway, below = 0.4] (1) {$\dots$};

    \draw (\leftspace+8, 4*\stepheight) -- ++ (2*\leftspace, 0)
        node[midway, above = 0.5] (1) {$\dbuf_{1,j}$}
        node[midway, below = 0.1] (1) {$m_j-Km_j^{\text{old}}$};

    \draw (\leftspace+10, 4*\stepheight) -- ++ (2*\leftspace, 0)
        node[midway, above = 0.5] (1) {$\dbuf_{2,j}$}
        node[midway, below = 0.4] (1) {$\dots$};

    \draw (\leftspace+12, 4*\stepheight) -- ++ (4*\leftspace, 0)
        node[midway, above = 0.5] (1) {$P$};

    \draw [->, color=red, thick] (\midpoint, 4.4 * \stepheight) -- ++ (0, 0.25 * \stepheight);

    \node[color=red] at (\midpoint+1, 4.55 * \stepheight) {logically};

    \draw (\leftspace-0.5, 5.7*\stepheight) -- ++ (0, 0)
        node[midway, above = -0.1] (1) {$(iii)$};
    
    \drawsegment[\leftspace,5.7*\stepheight,12,5,,,upper_final,lower_final];
    \drawsegment[\leftspace+12,5.7*\stepheight,4,0,,,upper_final,lower_final];

    \draw (\leftspace, 5.7*\stepheight) -- ++ (2*\leftspace, 0)
        node[midway, above = 0.57] (1) {segment $1$}
        node[midway, below = 0.1] (1) {$m_j^{\text{old}}$};

    \draw (\leftspace+2, 5.7*\stepheight) -- ++ (2*\leftspace, 0)
        node[midway, above = 0.57] (1) {segment $2$}
        node[midway, below = 0.1] (1) {$m_j^{\text{old}}$};
    \draw (\leftspace+8, 5.7*\stepheight) -- ++ (2*\leftspace, 0)
        node[midway, above = 0.57] (1) {$\cdots$};
    \draw (\leftspace+4, 5.7*\stepheight) -- ++ (2*\leftspace, 0)
        node[midway, above = 0.5] (1) {$\dbuf_{1,j}$}
        node[midway, below = 0.1] (1) {$m_j-Km_j^{\text{old}}$};

    \draw [decorate,decoration= {brace,amplitude=5pt}]
    (\leftspace, 5.2*\stepheight) -- ++ (6*\leftspace, 0) node[midway,above = 0.4]{phase-$j$ segment $1$};

    \draw (\leftspace+12, 5.7*\stepheight) -- ++ (4*\leftspace, 0)
        node[midway, above = 0.5] (1) {$P$};

    \draw [decorate,decoration= {brace,amplitude=5pt}]
    (\leftspace+6, 5.2*\stepheight) -- ++ (6*\leftspace, 0) node[midway,above = 0.4]{phase-$j$ segment $2$};

\end{tikzpicture}}
    \caption{An illustration of the merge subroutine with $B_{j-1}=4$ and $K=2$. (i) Right after the re-allocation subroutine, each phase-$(j-1)$ segment has remaining capacity $m_j^{\text{old}}$. (ii) In the merge subroutine, a dampening buffer is allocated for each phase-$j$ segment. (iii) Logically, the previous buffers of a phase-$j$ segment and the dampening buffer of that segment are functionally equivalent to one single buffer.}
\end{figure}
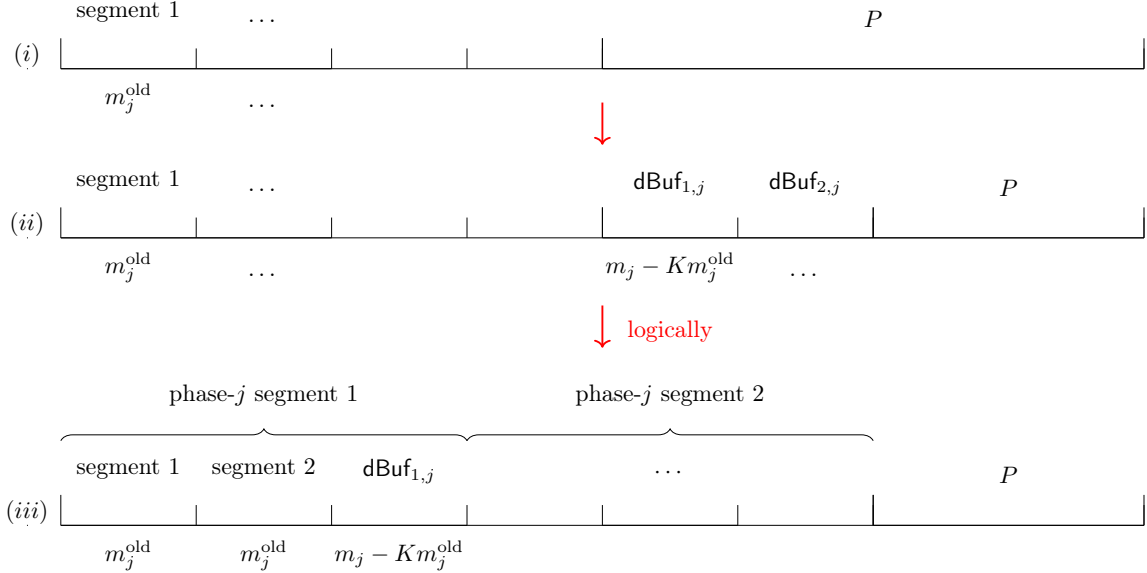

\subsection{The Analysis}

\paragraph{Proof of correctness.} We start by analyzing the behavior of dampening buffers. We have the following lemma, which is consistent with the intuitions in the technical overview. In particular, \cref{lem:warm_up_2_damp_order} implies that right after a merge subroutine, among all the buffers belonging to a current segment, the dampening buffer is the last to be full (with high probability).

\begin{restatable}{lemma}{dampOrder}
    \label{lem:warm_up_2_damp_order}
    We say that buffer $\buf_1$ \defn{precedes} buffer $\buf_2$, if $\buf_1$ is allocated before $\buf_2$, and the initial segment of $\buf_1$ is contained in that of $\buf_2$. For the algorithm $\algAdaptMerge$, with probability $1-O(1/n^{10})$, the following holds: When a buffer becomes full, all the buffers that precede it are already full.
\end{restatable}

The formal proof of \cref{lem:warm_up_2_damp_order} is deferred to \cref{app:warm_up_2}. We will be needing the full power of \cref{lem:warm_up_2_damp_order} in \cref{sec:final}, where we study the distribution of insertions to the buffers. For now, the importance of \cref{lem:warm_up_2_damp_order} is that, it implies the following corollary. \cref{cor:warm_up_2_buffer_property} essentially says that, at any point during the execution of $\algAdaptMerge$, for each current segment, we can think of the set of buffers that belong to it as functionally equivalent to one big buffer. This allows us to analyze buffers in the same way as in $\algAdapt$.

\begin{corollary}
    \label{cor:warm_up_2_buffer_property}
    With probability $1-O(1/n^{10})$, the following property holds for $\algAdaptMerge$: When a new item $x$ from current segment $i$ arrives, as long as the remaining capacity $c_i$ is positive, $x$ will be inserted into a buffer belonging to $i$.
\end{corollary}

\begin{proof}[Proof of \cref{cor:warm_up_2_buffer_property} using \cref{lem:warm_up_2_damp_order}]
    When $x$ arrives, consider the current youngest buffer $\buf$ belonging to $i$. By the execution of our algorithm, the initial segment of $\buf$ must be segment $i$. Therefore, all other buffers belonging to $i$ precede $\buf$. If $c_i>0$, then we can apply \cref{lem:warm_up_2_damp_order} and conclude that, $\buf$ cannot be full at this moment. Therefore, even if all the previous non-full buffers don't accept $x$, $\buf$ definitely will. That is, $x$ will be inserted into a buffer.
\end{proof}

It is then easy to prove the correctness of $\algAdaptMerge$, by using an argument similar to that of $\algAdapt$. Here, we present an informal proof: Perform induction on the phases. For any phase $j$, we start at a state similar to a phase in $\algAdapt$, that is, we have $B_j$ segments, and each segment has remaining capacity $m_j$. Given \cref{cor:warm_up_2_buffer_property}, we have that the third type of failure will not happen (with high probability) until the last phase. Therefore, the condition of $\algAdaptMerge$ failing during some phase $j$ is similar to that in $\algAdapt$, and can be checked by analyzing the number of items from each phase-$j$ segment that arrive during phase $j$. This analysis is very similar to \cref{clm:warm_up_phase_concentration}. We can thus show that

\begin{restatable}{lemma}{AdaptMergeSuccess}
    \label{lem:warm_up_2_regularity}
    The probability that $\algAdaptMerge$ fails is $O(1/n^{10})$.
\end{restatable}

The formal proof of \cref{lem:warm_up_2_regularity} is also deferred to \cref{app:warm_up_2}.

\paragraph{Proof of efficiency.} Given that $\algAdaptMerge$ succeeds, we have the following properties.

\begin{lemma}
    \label{lem:warm_up_2_decrement}
    For sufficiently large $n$, conditioned on $\algAdaptMerge$ succeeding, for any phase $j$, $n_j/B_j$ is in the range $[K,10K^{2.5}]$. Moreover, the number of phases is $O(\log n/\log\log n)$.
\end{lemma}

Note that the sizes of the buffers allocated at phase $j$ are at most $m_j$, which is at most $n_j/B_j$. Thus \cref{lem:warm_up_2_decrement} implies that, all the buffers in $\algAdaptMerge$ have size $\poly\log n$. This is very different from $\algAdapt$, in which the buffer sizes could be as large as $n^{\Omega(1)}$.

The behavior of buffer sizes here is very different from $\algAdapt$. In $\algAdapt$, the sizes of buffers decrease rapidly, where the sizes of the phase-$j$ buffers are at most $(n/B)^{(2/3)^{j-1}}$. In comparison, here we keep the buffer sizes small throughout the algorithm, just barely large enough for applying the Chernoff bound.

\begin{proof}
    We prove the results one by one.
    \begin{itemize}
        \item We prove $n_j/B_j\ge K$ by induction: In phase $1$, this holds by definition of $B_1$. In phase $j>1$, since $n_{j-1}/B_{j-1}\ge K$, we have that the pool size during phase $j-1$ is $n_{j-1}-B_{j-1}\cdot m_{j-1}=\Omega(B_{j-1}\cdot (n_{j-1}/B_{j-1})^{2/3})=\omega(B_{j-1})$. Note that $n_j$ must be at least the size of this pool, which means that $n_j/B_{j-1}\ge 1$. If $n_j/B_{j-1}\ge K$, then we are done because $B_j\le B_{j-1}$. Otherwise, if $n_j/B_{j-1}<K$, then we will perform a merge subroutine at phase $j$, which means that $B_j=B_{j-1}/K$. Since $n_j/B_{j-1}\ge 1$, we have that $n_j/B_j\ge K$.
        \item To show that $n_j/B_j=O(K^{2.5})$ for any phase $j$, we let $j$ be the phase that maximizes $n_j/B_j$. If $j=1$, then by our definition of $B_1$ (i.e., $B_1$ is the largest power of $K$ such that $B_1\cdot K\le n$), we have that $B_1\cdot K^2>n$, so the bound holds.
    
        If $j>1$, then we must have $n_j/B_j>n_{j-1}/B_{j-1}$. This happens only when the merge subroutine is called, and $B_j$ is set to be $B_{j-1}/K$. By the condition for starting a merge subroutine, the size of the pool (after the adaptive allocation of phase $j$) is smaller than $B_{j-1}\cdot K$. On the other hand, the size of the pool is equal to $n_{j}-B_{j-1}\cdot m_{j}^{\text{old}}$ by definition, which, by expanding the definition of $m_{j}^{\text{old}}$, is at least $B_{j-1}\cdot d(n_{j}/B_{j-1})\ge B_{j-1}\cdot ((1/2)(n_j/^{B_{j-1}})^{2/3}-1)$. 
        
        Comparing the two above bounds for the size of the pool, we have that $(n_{j}/B_{j-1})^{2/3}\le 2K+2\le 3K$ ($n$ is sufficiently large). Since $B_j=B_{j-1}/K$, we have that $(n_j/B_j)=(n_j/B_{j-1})\cdot K\le 10K^{2.5}$.
        \item As for the number of phases, we use the fact that the first type of failure does not happen, i.e., 
        \begin{align*}
            n_j/B_{j-1}\le (n_{j-1}/B_{j-1})^{2/3}
        \end{align*}
        for any $j>1$. By rearranging the terms, we have that
        \begin{align*}
            n_j\le n_{j-1} /(n_{j-1}/B_{j-1})^{1/3}.
        \end{align*}
        Combining this with $n_{j-1}/B_{j-1}\ge K$ (the first result of this lemma) implies that $n_j\le n_{j-1}/K^{1/3}$. Since $K=\Theta(\log^8 n)$, there can only be $O(\log n/\log\log n)$ phases.
    \end{itemize}

\end{proof}

We now prove \cref{thm:warm_up_2_high_prob}.

\begin{proof}[Proof of \cref{thm:warm_up_2_high_prob}]
    Conditioned on $\algAdaptMerge$ succeeding, its cost can be analyzed as follows:

    \begin{itemize}
        \item \textbf{The cost of the buffers:} Fix a phase $j$, and consider all the regular buffers allocated at the start of phase $j$. The cost of each buffer is at most
        \begin{align*}
            (\text{buffer length})/B_{j-1}
        \end{align*}
        in the worst case. Since there are $B_{j-1}$ such buffers, their total cost is bounded by their maximum length, which is $O(\log^{20}n)$. A similar analysis works for the dampening buffers of phase $j$. Since there are $O(\log n/\log\log n)$ phases, the total cost of all buffers is $O(\log^{21}n/\log\log n)$.
        \item \textbf{The cost of the final pool:} This is at most the size of the final pool, which by definition is $O(K^{2.5})=O(\log^{20} n)$.
        \item \textbf{The inter-subarray cost:} Similar to the regular buffers, we have that the regular buffers of each phase are arranged in increasing order of their segment. The same holds for the dampening buffers. Therefore, the inter-subarray cost of the buffers of the same phase is $O(1)$, and the total inter-subarray cost is $O(\log n/\log\log n)$.
    \end{itemize}

    Therefore, the competitive ratio of $\algAdaptMerge$ is $\poly\log n$ with high probability.
\end{proof}

\section{The Final Algorithm: Recursing on Non-Uniform Subproblems}
\label{sec:final}

\mainTheorem*

In this section, we further improve upon $\algAdaptMerge$, by running it recursively on the buffers and the final pool. This recursion is much harder to analyze than $\algAdapt$, because here the items inserted into a buffer are not uniformly randomly distributed.

In order to run $\algAdaptMerge$ recursively, we first need to analyze the distribution of insertions to each buffer. It turns out that, the regular buffers receive (approximately) uniformly random items sampled from its initial segment, which means that they are similar to the buffers in $\algAdapt$, and we can run $\algAdaptMerge$ recursively on them. 

As for the dampening buffers, we show that the distribution of insertions to them are similar to that of a pool, i.e., the distribution is close to a multi-way pool distribution with $K$ segments. We then run $\algAdaptMerge$ recursively on a dampening buffer, by thinking of it as a \emph{pool}. That is, we use the knowledge about the buffers preceding a dampening buffer, to adaptively allocate new buffers within the dampening buffer.

\subsection{Regular Buffers Receive Random Items}

Given \cref{lem:warm_up_2_damp_order}, the distribution of insertions to each regular buffer can be shown to be (almost) random. \cref{lem:warm_up_2_damp_order} implies that, after allocating a regular buffer $\buf_{i,j}$, the set of items inserted to it (with high probability) only depends on whether each future item is from segment $i$, and not on the actual values of these items. Therefore, the items inserted into $\buf_{i,j}$ should be uniformly random.

\begin{lemma}
    \label{lem:dist_reg_buffer}
    Fix $1\le l\le n$. For $\algAdaptMerge$, with probability $1-O(1/n^{5})$ over a random prefix of the input $x_1,\dots,x_l$ of length $l$, the following holds: For any regular buffer $\buf_{i,j}$ that was allocated right after receiving $x_l$, the distribution of items inserted to $\buf_{i,j}$ (over the randomness of the future items $x_{l+1},\dots,x_n$) is $O(1/n^{5})$-close to being uniformly randomly distributed in segment $i$.
\end{lemma}

We say that two distributions are $\alpha$-close, if their statistical distance is at most $\alpha$.

\begin{proof}
    We say that a prefix $x_1,\dots,x_l$ of input is \defn{good}, if after we complete the prefix with random items $x_{l+1},\dots,x_n$, the resulting input sequence satisfies the properties in both \cref{lem:warm_up_2_damp_order} and \cref{lem:warm_up_2_regularity} with probability $1-O(1/n^{5})$ over the randomness of $x_{l+1},\dots,x_n$.

    We know that a random input $x_1,\dots,x_n$ satisfies the properties in both \cref{lem:warm_up_2_damp_order} and \cref{lem:warm_up_2_regularity} with probability $1-O(1/n^{10})$. Using Markov's inequality, for a fixed $l$, the probability that a prefix $x_1,\dots,x_l$ is good is $1-O(1/n^{5})$.
    
    We now show that \cref{lem:dist_reg_buffer} holds deterministically for any good prefix. Fix a regular buffer $\buf_{i,j}$ that is allocated right after receiving $x_l$, whose initial segment is phase-$(j-1)$ segment $i$. After the allocation of $\buf_{i,j}$, there are $n_{j}=n-l$ items in the future, and the remaining capacity of segment $i$ (which is equal to $|\buf_{i,j}|$ plus the remaining capacities of the buffers that precede $\buf_{i,j}$) is $m_{j}^{\text{old}}$.
    
    Since $x_1,\dots,x_l$ is good, we have that: With probability $1-O(1/n^{5})$ over the future $n_{j}=n-l$ items (which are randomly distributed), the first $m_{j}^{\text{old}}-|\buf_{i,j}|$ future items from segment $i$ would be inserted into buffers that precede $\buf_{i,j}$, which would fill all those buffers. 
    
    Conditioned on this, the set of items inserted to $\buf_{i,j}$ would exactly be the $(m_{j}^{\text{old}}-|\buf_{i,j}|+1)$-th to the $m_{j}^{\text{old}}$-th future items that take value from segment $i$. These items are (approximately) uniformly randomly distributed in segment $i$.
\end{proof}

\subsection{Dampening Buffers are Similar to Pools}

We now prove a similar lemma for the dampening buffers. Intuitively, dampening buffers behave like pools. For a dampening buffer $\dbuf_{i,j}$ When a new item $x$ from segment $i$ comes, it first attempts to enter any buffer that precedes $\dbuf_{i,j}$, and used to belong to the same phase-$(j-1)$ segment as $x$. Only when those buffers are full will the new item enter $\dbuf_{i,j}$. That is, $\dbuf_{i,j}$ behaves like a ``pool'' for the $K$ phase-$(j-1)$ segments. We now formalize this intuition.

\begin{lemma}
    \label{lem:dist_merger_buffer}
    Fix $1\le l\le n$. For $\algAdaptMerge$, with probability $1-O(1/n^{5})$ over a random prefix $x_1,\dots,x_l$, the following holds: For any dampening buffer $\dbuf_{i,j}$ that is allocated right after receiving $x_l$, suppose that after the allocation of $\dbuf_{i,j}$. Then with probability $1-O(1/n^{5})$, the insertions to $\dbuf_{i,j}$ are as follows: An item $x$ is inserted into $\dbuf_{i,j}$, if and only if $x$ is from segment $i$, and between the allocation of $\dbuf_{i,j}$ and the arrival of $x$, we have seen at least $m_j^{\text{old}}$ items from the same phase-$(j-1)$ segment as $x$. 
\end{lemma}

In other words, the insertions to $\dbuf_{i,j}$ approximately follows a ``scaled'' version of the multi-way pool distribution $\mathcal{P}_{m_j,K,m_j^{\text{old}}}$. It's a scaled version because the items are not from $[0,1]$, but from segment $i$.

\begin{proof}
    Similar to \cref{lem:dist_reg_buffer}, we only consider the good prefixes $x_1,\dots,x_l$. For a dampening buffer $\dbuf_{i,j}$ that is allocated right after receiving $x_l$, its initial segment is phase-$j$ segment $i$, which consists of $K$ phase-$(j-1)$ segments. We denote these phase-$(j-1)$ segments as $S_1,\dots,S_K$.
    
    The buffers that precede $\dbuf_{i,j}$ can be divided into $K$ groups $\mathcal{B}_1,\dots, \mathcal{B}_K$, partitioned based on the phase-$(j-1)$ segment that they belonged to. When $\dbuf_{i,j}$ is allocated, each group $\mathcal{B}_k$ has total remaining capacity $m_j^{\text{old}}$. Consider the youngest buffer in $\mathcal{B}_k$, denoted as $\buf_k$. Due to the execution of the algorithm, the initial segment of $\buf_k$ must be $S_k$. This means that all other buffers in $\mathcal{B}_k$ precede $\buf_k$.

    Since the prefix is good, we can apply \cref{lem:warm_up_2_damp_order} on the buffers $\buf_k$, and conclude that: With high probability, after the allocation of $\dbuf_{i,j}$, among the future items, the buffers in $\mathcal{B}_k$ will be filled by the first $m_j^{\text{old}}$ items that are from $S_k$. Then, $\dbuf_{i,j}$ will be filled by the first items that are from segment $i$ ($=S_1\cup\dots\cup S_K$), and are not inserted into the older buffers. This proves the lemma.
\end{proof}

Additionally, since dampening buffers are similar to the pool, it is natural that the final pool also satisfies a property similar to \cref{lem:dist_merger_buffer}.

\begin{lemma}
    \label{lem:dist_pool}
    Fix $1\le l\le n$. With probability $1-O(1/n^{4})$ over a random prefix $x_1,\dots,x_l$, the following holds with probability $1-O(1/n^{5})$: If the last phase starts right after receiving $x_l$, then the items inserted into the final pool are generated as follows: A new item $x$ is inserted into the final pool, if and only if we have seen at least $m_j$ items that arrive during phase $j$ and before $x$, that are from the same phase-$j$ segment as $x$.
\end{lemma}

In other words, the insertions to the final pool approximately follows the multi-way pool distribution $\mathcal{P}_{n_j,K,m_j}$. The proof is similar to \cref{lem:dist_merger_buffer}, and is omitted.

\subsection{The Final Algorithm $\algFinal$}

Here we present the final algorithm $\algFinal$. Other than the recursion, the only difference between $\algFinal$ and $\algAdaptMerge$ is that, we set $K$ to be of the form $2^{2^k}$ satisfying $K\in (\log^8 n,\log^{16} n]$, where $k$ is an integer. Note that such $K$ must exist. This is to make it easier for us to merge segments during recursion.

Within each regular buffer, we just recursively apply a scaled version of $\algFinal$. Since they receive (approximately) uniformly random items, their cost will be small.

For the dampening buffers and the final pool, our recursion is more complicated. This is handled by a special recursive subroutine $\algPool$, described below. We first define the interface for $\algPool$, and discuss its implementation later.

\paragraph{Interface of $\algPool$.} Intuitively, $\algPool$ is used in the following way: When running $\algPool$ to recursively handle the insertions to a dampening buffer $\buf$, the inputs given to $\algPool$ consist of both the items inserted to $\buf$, and the ones inserted to the buffers \emph{preceding} $\buf$. This is because, since we think of a dampening buffer as a pool, we need to make use of the knowledge of the items entering ``its buffers'' (i.e., the buffers preceding $\buf$), to perform adaptive allocations within $\buf$.

Formally, the recursive algorithm $\algPool$ takes three parameters: The number of items $n'$, the number of segments $B'$, and the initial buffer size $m'$. We require that, $B'$ is of the form $2^{2^k}$ for some integer $k$ (when we apply $\algPool$, $B'$ will be the parameter $K$ used by the outer algorithm for merging, which satisfies this property), and that $n'-B'\cdot m'\ge 0$.

$\algPool$ is responsible for an array of $n'-B'\cdot m'$ cells (i.e., the dampening buffer). However, its input is $n'$ \emph{uniformly random} items from $[0,1]$. 

$\algPool$ is required to do the following: Partition $[0,1]$ into $B'$ segments. When an item $x$ arrives, if prior to $x$, we have already seen $m'$ items from the same segment as $x$, then $\algPool$ needs to place $x$ in one of its empty cells (we call such $x$ a \defn{real item}); Otherwise, $\algPool$ need not do anything for $x$ (we call such $x$ an \defn{informative item}, since the only reason that such item is in the input is to inform $\algPool$ of its existence).

\paragraph{Recursion for the dampening buffers (and the final pool).} For a dampening buffer $\dbuf_{i,j}$, we run a (scaled) instance of $\algPool$, with $m_j$ items, $K$ segments, and initial buffer size $m_j^{\text{old}}$. Its input are the first $m_j$ items belonging to phase-$j$ segment $i$, that arrive after the allocation of $\dbuf_{i,j}$.

Formally, the outer algorithm does the following: Denote phase-$j$ segment $i$ as $[v_l,v_r]$. To handle $\dbuf_{i,j}$, the outer algorithm first initiates an instance of $\algPool(m_j,K,m_j^{\text{old}})$. Whenever an item $x$ from phase-$j$ segment $i$ arrives, the outer algorithm sends an item of value $(x-v_l)/(v_r-v_l)$ to the recursive instance $\algPool$. This scaling is to make sure that $\algPool$ receives uniformly random items from $[0,1]$. If $\algPool$ decides that this item is a real item, and places it in the $k$-th cell in its array, then the outer algorithm places $x$ in the $k$-th cell in $\dbuf_{i,j}$.

Note that given \cref{lem:dist_merger_buffer}, when $\algPool(m_j,K,m_j^{\text{old}})$ is given this (scaled) input, with high probability, the items that the outer algorithm inserts to $\dbuf_{i,j}$ are exactly the \emph{real items} that our scaled instance of $\algPool(m_j,K,m_j^{\text{old}})$ receives. If this does not hold, then we enter the failure mode, and allow the outer algorithm to have cost $n$.

For the final pool, we also run a (scaled) instance of $\algPool$, with $n_j$ items, $K$ segments, and initial buffer size $m_j$ (where $j$ is the last phase). Its input is all the items that arrive during phase $j$. The inputs to the final pool need not be scaled, since they are already sampled uniformly from $[0,1]$.

This concludes the implementation of $\algFinal$. We now turn to the recursive algorithm $\algPool$.

\paragraph{Notations for $\algPool$.} Similar to $\algAdaptMerge$, $\algPool$ also runs in phases. We use $n'_j,B'_j$ and $(m')_j^{\text{old}},m'_j$ to denote the number of future (real and informative) items during phase $j$, the number of segments during phase $j$ and the lower bounds for number of future items in each segment, respectively. $K'$ is a parameter used in merging, which is defined to be $2^{2^{k'}}$ for some (unique) integer $k'$, satisfying $K'\in (\log^8 (n'),\log^{16}(n')]$. This definition is to make sure that, as long as $B'\ge K'$, we must have that $B'$ is a power of $K'$. 

For phase $1$, we have that $n'_1=n'$, $B'_1=B'$ and $m'_1=m'$. Note that this parametrization is slightly different from $\algAdaptMerge$, in that $m'_1=(n'_1/B'_1)-d(n'_1/B'_1)$ does not necessarily hold. We will merge segments in phase $1$, so we need not define $(m')_1^{\text{old}}$.

To characterize the idea that we have $B'$ ``imaginary buffers'', we additionally define a notion of \defn{budget}, for each phase-$1$ segment. Initially, the budget of any phase-$1$ segment is $m'$. Whenever an informative item arrives, the budget of its phase-$1$ segment is decreased by $1$. At any point, the \defn{remaining capacity of a current segment} (not necessarily phase-$1$) is then redefined to be the total remaining capacity of the buffers belonging to it, \emph{plus} the current budget of all the phase-$1$ segments contained in it.

\paragraph{Implementation of $\algPool$.}

At the start of phase $1$, we \emph{do not} allocate buffers, and the phase-$1$ pool is defined to be our entire array of length $n'-B'\cdot m'$. All the items that arrive in phase $1$ are informative items. Phase $1$ ends when the remaining capacity of some phase-$1$ segment becomes $0$, or equivalently, when the budget of some phase-$1$ segment becomes $0$.

For the other phases, our behavior is similar to $\algAdaptMerge$: For $j>1$, we also define $m'_j=(n'_j/B'_j)-d(n'_j/B'_j)$ and $(m')_j^{\text{old}}=(n'_j/B'_{j-1})-d(n'_j/B'_{j-1})$. The rules for allocating new buffers and deciding whether to run a merge subroutine, as well as the criteria for entering the failure mode, are the same as $\algAdaptMerge$. The only ``difference'' is that, the third type of failure (i.e., we fail if we have no place for a new item) only happens for \emph{real} items. Again, the definition for success only consider the $\algPool$ instance itself, not its subproblems.

When a real item arrives, we select a subarray for it in the same way as $\algAdaptMerge$. That is, the item is inserted into the first non-full buffer that accepts it. If there is no such buffer, then it's inserted into the (final) pool.

Finally, within the buffers allocated during an $\algPool$ instance (as well as the final pool), we also need to run a recursive algorithm. This is exactly the same as in $\algFinal$.

\subsection{The analysis}

\paragraph{The size of the recursive algorithms.} Before going into the analysis, we first ask the following important question: When an instance of $\algPool$ is called, what relationships do the parameters $n',B',m'$ satisfy? For this, we use the following lemma which is proven in \cref{app:warm_up_2}:

\begin{restatable}{lemma}{mergerConcentration}
    \label{lem:merger_concentration}
    In $\algAdaptMerge$, when $n$ is sufficiently large, for any phase $j$, we have that:
    \begin{itemize}
        \item If we performed merging at the start of phase $j$, then
        \begin{align*}
            m_j^{\text{old}}= \frac {m_j}K-(1\pm o(1))\cdot d\bk*{\frac {m_j}K}.
        \end{align*}
        \item For any phase $j$, $m_j^{\text{old}}\ge (1/5)K^{2/3}=\omega(\log ^4n)$. Combining this with the previous bound implies that $m_j\ge (1/5)K^{5/3}$.
    \end{itemize}
\end{restatable}

In addition to the properties in \cref{lem:merger_concentration}, using \cref{lem:warm_up_2_decrement}, we also have that $m_j\le 10K^{2.5}$. We can thus conclude that, whenever an instance of $\algPool$ is called for a dampening buffer of some phase $j$, for the parameters $n',B',m'$ (which correspond to $m_j,K,m_j^{\text{old}}$, respectively), we have that $B'\in [C_1(n')^{2/5},C_2(n')^{3/5}]$ and $m'=(n'/B)-(1\pm o(1))\cdot d(n'/B')$ where $C_1,C_2$ are absolute constants. The same also holds for the final pool, where $n',B',m'$ correspond to $n_j,K,m_j$ ($j$ is the last phase), respectively. 

In the following, if an $\algPool$ instance satisfies that $m'=(n'/B)-(1\pm o(1))\cdot d(n'/B')$ and $B'\in [C_1(n')^{2/5},C_2(n')^{3/5}]$, we say that it is \defn{well-parametrized}. To compare this with the previous parameters, in $\algAdaptMerge$, we have that $B_1=n/\poly\log n$, and $m_1=(n_1/B_1)-d(n_1/B_1)$. This difference affects our analysis:

\begin{itemize}
    \item $B'$ is no longer $(n')/\poly\log (n')$, which means that the sizes of the phase-$1$ buffers could be much larger than $\poly\log (n')$. Therefore, in the first phases of $\algPool$, its behavior is more similar to $\algAdapt$ than to $\algAdaptMerge$. That is, the number of segments does not change, and the buffer sizes decrease polynomially. This goes on for $O(\log\log (n'))$ phases, until the buffer sizes are small enough (i.e., $\poly\log (n')$). After this, $\algPool$ will behave like $\algAdaptMerge$, where there are $O(\log(n')/\log\log(n'))$ more phases, and the new buffers have sizes $\poly\log (n')$.
    \item The fact that $m'=(n'/B)-(1\pm o(1))\cdot d(n'/B')$ turns out to be only a small inconvenience: We only have to prove that, with high probability, $\algPool$ will not fail in phase $1$. For the other phases, we will make sure that $m'_j=(n'_j/B'_j)-d(n'_j/B'_j)$, which means that the future phases can be analyzed in the same way as in $\algAdapt$ and $\algAdaptMerge$.
\end{itemize}

Finally, we remark that the same relationships hold for \emph{all} the instances of $\algPool$, not just the ones called by an $\algFinal$ instance. That is, when we allocate a dampening buffer within a $\algPool$ instance, \emph{its} recursive instance of $\algPool$ is also well-parametrized. This can be proven in the same way as the previous bounds, so the formal proofs are omitted.

\paragraph{Proof of correctness.} We first show that the recursive algorithms also succeed with high probability. Note that, the high success probability is with respect to the size of the specific recursive sub-problem, not $n$. Here, we assume that each recursive subproblem receives perfectly uniformly random items.

For regular buffers, since they receive approximately uniformly random items by \cref{lem:dist_reg_buffer}, we can simply apply \cref{lem:warm_up_2_regularity}. 

For the instances of $\algPool$, we prove the following:

\begin{lemma}
    Fix any well-parametrized instance of $\algPool$. If it receives $n'$ uniformly random items, then it succeeds with probability $1-O(1/(n')^{10})$.
\end{lemma}

\begin{proof}
    We only prove the lemma for sufficiently large $n'$. The proof is similar to the proof of \cref{lem:warm_up_regularity} and \cref{lem:warm_up_2_regularity}, and is by induction on the phases. Here, we only show that the algorithm will not fail in phase $1$, because it is the only phase where we don't have $m'_1=(n'_1/B'_1)-d(n'_1/B'_1)$. For the future phases, due to the execution of the recursive algorithm, this equality will hold, and the proof for those phases is similar to \cref{lem:warm_up_regularity} (for the phases before the first merge subroutine) and \cref{lem:warm_up_2_regularity} (for the phases after the first merge subroutine).

    For phase $1$, since we do not receive any real item, it is by definition that the third type of failure will not happen. Also, since the buffer sizes are large, we won't merge segments at the end of phase $1$, which means that for the second type of failure, we only have to worry about the regular buffers. Therefore, to show that the recursive algorithm does not fail, it remains to show that the first two types of failure will not happen. That is:
    \begin{enumerate}
        \item When phase $1$ ends, the number of future items $n'_2$ is at most $B'_1\cdot (n'_1/B'_1)^{2/3}$.
        \item When phase $1$ ends, for any segment $i$, we have that the remaining capacity of $i$ is at most $m'_2$.
    \end{enumerate}
    A slightly modified version of \cref{clm:warm_up_phase_concentration} (which we do not state) shows that, when phase $1$ ends, with high probability the remaining capacity of each segment is at most $(1/5)\cdot d(n'_1/B'_1)$. To bound $n'_2$, we use the fact that it is at most the sum of the pool size and the remaining capacities. That is,
    \begin{align*}
        n'_2\le (n'_1-B'_1\cdot m'_1)+B'_1\cdot (1/5)\cdot d(n'_1/B'_1).
    \end{align*}
    Since $m'_1=(n'_1/B'_1)-(1\pm o(1))\cdot d(n'_1/B'_1)$, we can show that the right hand side is at most $B'_1\cdot (n'_1/B'_1)^{2/3}$, by following the proof in \cref{lem:warm_up_regularity}.

    To show that the remaining capacities are small, similar to the proof of \cref{lem:warm_up_regularity}, we have that $n'_2$ is at least $n'_1-B'_1\cdot m'_1$, which is $B'_1\cdot d(n'_1/B'_1)\cdot (1\pm o(1))$. Since $m'_2=(n'_2/B'_1)-d(n'_2/B'_1)\ge (1/2)\cdot (n'_2/B'_1)$, we can conclude that $m'_2\ge (1/5)\cdot d(n'_1/B'_1)$ just as before.
\end{proof}

\paragraph{Proof of efficiency.} Given the correctness, we now prove the cost of the algorithm by induction. Note that \cref{lem:real_recurse} straightforwardly implies \cref{thm:real_exp}, because the algorithm for the entire array is an instance of $\algFinal$.

\begin{lemma}
    \label{lem:real_recurse}
    For any $\algFinal$ instance of length $n$, as well as any well-parametrized $\algPool$ instance with $n'=n$, the expected cost is $\log n\cdot 2^{O(\log^* n)}$.
\end{lemma}

\begin{proof}
    Define $g(n)$ to be the number of times we can apply $n\mapsto \log^{41}n$ before $n$ becomes less than $2$. We can prove by induction that $g(n)=\Theta(\log^*n)$. In the following, we use $g(n)$ as a proxy of $\log^* n$.
    
    Let $f(n)=C\cdot \log n\cdot 2^{C\cdot g(n)}$ for some sufficiently large constant $C$. In the following, all the big-O notations do not hide $C$. We prove a slightly modified version of \cref{lem:real_recurse}, where we show that any $\algFinal$ instance of size $n$ has expected cost at most $f(n)$, and any well-parametrized $\algPool$ instance of size $n$ has expected cost at most $10f(n)$. The proof is by induction.

    In the base case, let $n_0$ be a sufficiently large constant. We let $C$ be large enough, so that $f(n)\ge n$ for any $n\le n_0$. In the following, we only consider $n>n_0$.
    
    Now assume that for some $n>n_0$, the induction hypothesis holds for any $n'<n$, we show that the induction hypothesis also holds for $n$. For an $\algFinal$ instance of length $n$, the recursive algorithm is similar to $\algAdaptMerge$. Let $K$ denote the parameter used for merging, i.e., $K=2^{2^{k}}$ for some integer $k$, and $K\in (\log^8 n,\log^{16}n]$. Similar to \cref{lem:warm_up_2_decrement} (the parameter $K$ is slightly different here, but this does not affect the analysis by too much), during its execution, all the allocated buffers (as well as the final pool) would have size at most $O(K^{2.5})=O(\log^{40}n)$. The cost of an $\algFinal$ instance can thus be analyzed as follows:
    \begin{itemize}
        \item \textbf{The cost of the buffers (and the final pool):} Similar to $\algAdaptMerge$, $\algFinal$ has $O(\log n/\log\log n)$ phases. During each phase, we have that $m_j=O(K^{2.5})$, which means that the size of any recursive subproblem is at most $O(K^{2.5})=O(\log^{40}n)$ (recall that, the size of an instance of $\algFinal$ is the buffer size, and the size of an instance of $\algPool$ is $n'=m_j$). Using the induction hypothesis, the cost for this part is $O(\log n/\log\log n)\cdot f(O(K^{2.5}))$. Since $g(n)\ge g(O(K^{2.5}))+1$ by definition of $g$, we have that      
        \begin{align*}
            O(\log n/\log\log n)\cdot f(O(K^{2.5})) &\le  O(\log n/\log\log n)\cdot C\cdot \log(O(\log^{40}n))\cdot 2^{C\cdot g(n)-C} \\
            &\le O(1)\cdot \log n\cdot 2^{Cg(n)}\cdot C\cdot 2^{-C} \\
            &=O(f(n))\cdot 2^{-C}.
        \end{align*}
        When $C$ is a sufficiently large constant, we have that this part is at most $(1/2)\cdot f(n)$.
        \item \textbf{The inter-subarray cost:} This is $O(\log n/\log\log n)$ as before.
        \item \textbf{Cost of failure:} In the case that $\algFinal$ itself enters the failure mode, the expected cost is $o(1)$; 
        \item \textbf{Cost of imperfect randomness:} In the case that it succeeds, the distribution of inputs to each recursive subproblem is only \emph{close} to being uniformly random, and not perfectly uniformly random as we assumed. This also incurs $o(1)$ cost.
    \end{itemize}

    Next, fix a well-parametrized $\algPool$ instance with parameters $n'=n,B',m',K'$. To bound the cost, we use a combination of analyses for both $\algAdapt$ and $\algAdaptMerge$ (\cref{lem:warm_up_decrement} and \cref{lem:warm_up_2_decrement}), and show that: Until the size of the pool becomes smaller than $B'_1\cdot K'$, we will not be merging segments, and the algorithm behaves similarly to $\algAdapt$. That is, the buffer sizes of each phase are at most $ (n'_1/B'_1)^{2/3},(n'_1/B'_1)^{4/9},\dots$, respectively (recall that we don't really have buffers in phase $1$). Note that it is crucial to our analysis that all these buffers are \emph{regular} buffers. Then, after the first merge subroutine, the size of buffers would always be at most $O((K')^{2.5})=\poly\log (n)$ (and might be dampening buffers), just like in $\algAdaptMerge$.

    Therefore, the buffer (and final pool) cost of a $\algPool$ instance consists of two parts, where the first part is 
    \begin{align*}
        &f(O(n^{(3/2)\cdot (2/3)}))+f(O(n^{(3/2)\cdot (4/9)}))+\cdots \\
        \le {}&C\cdot 2^{Cg(n)}\cdot [\log O(n)+\log O(n^{2/3}))+\cdots] \\
        \le {}&f(n)\cdot [1+(2/3)+\cdots+o(1)]\le 8f(n),
    \end{align*}
    and the second part is
    \begin{align*}
        (\log n/\log\log n)\cdot 10f((K')^{2.5})\le f(n). \tag{Similar to the $\algFinal$ instances.}
    \end{align*}
    The other costs are analyzed in the same way as $\algFinal$. Summing these costs up, the total expected cost is at most $10f(n)$.
\end{proof}

\section{Lower Bound for Stochastic Online Sorting}
\label{sec:lb}

In this section, we establish a lower bound of $\Omega(\log n)$ for stochastic online sorting. 

\mainLB*

The intuition behind this lower bound is that, regardless of the actions taken during the early stage, the last few items always ``incur'' a significant cost. For instance, consider the $i$-th-to-last item $x_{n-i+1}$. After we insert the first $n-i$ items, only $i$ empty cells remain. Suppose that none of the empty cells are adjacent to each other, so there are $\le 2i$ items that are neighbors of empty cells. With constant probability, the difference between item $x_{n-i+1}$ and any of the $2i$ neighbor items is at least $\Omega(1/i)$. In this case, no matter what the algorithm does, the cost of inserting $x_{n-i+1}$ is always $\Omega(1/i)$. Summing this over for all the $i$'s, we get an expected cost of $\Omega(\log n)$.

In general, it is difficult to define the cost of an item, because by the time we place an item, its neighbors may not be determined yet. To deal with this, in the formal proof, we will look at a random pair of adjacent cells, and bound their expected difference.

\begin{proof}
    Without loss of generality, we restrict our attention to deterministic algorithms, since the cost of a randomized algorithm is simply the expected cost over a distribution of deterministic ones. In the following discussion, we fix a deterministic algorithm.

    Let $x_1,\dots,x_n$ be a sequence of random inputs, and let $a_1,\dots,a_n$ be the final values of the cells \emph{after} processing all $n$ items. Our lower bound claim is equivalent to showing that
    \begin{align}
        \label{equ:goal}
        \E_{\{x_n\}}\Bk*{\sum_{j=1}^{n-1}\big|a_{j+1}-a_j\big|}=\Omega(\log n).
    \end{align}

    Let $1\le j\le n-1$ be a uniformly random index. We aim to show that, for any integer $1\le i\le n$,
    \begin{align}
        \label{equ:surrogate}
        \Pr_{\{x_n\},j}\Bk*{\big|a_{j+1}-a_j\big|\ge \frac 1{4i}}=\Omega\bk*{\frac in}.
    \end{align}
    If we can show \eqref{equ:surrogate}, then it implies \eqref{equ:goal}:
    \begin{align*}
        \E_{\{x_n\}}\Bk*{\sum_{j=1}^{n-1}\big|a_{j+1}-a_j\big|}&=(n-1)\cdot \E_{\{x_n\},j}\Bk*{\big|a_{j+1}-a_j\big|} \\
        &= (n-1)\cdot \int_0^{1}\Pr_{\{x_n\},j}\Bk*{\big|a_{j+1}-a_j\big|\ge y}\cdot\mathrm d y \\
        &> (n-1)\cdot \sum_{i=1}^{n}\int_{\frac 1{4(i+1)}}^{\frac 1{4i}}\Pr_{\{x_n\},j}\Bk*{\big|a_{j+1}-a_j\big|\ge y}\cdot \mathrm dy \\
        &\ge (n-1)\cdot \sum_{i=1}^{n}\Pr_{\{x_n\},j}\Bk*{\big|a_{j+1}-a_j\big|\ge \frac 1{4i}}\cdot \bk*{\frac 1{4i}-\frac 1{4(i+1)}} \\
        &=(n-1)\cdot \Omega\bk*{\sum_{i=1}^{n}\frac in\cdot \frac 1{i^2}}\tag{by \eqref{equ:surrogate}} \\
        &=\Omega\bk*{\sum_{i=1}^{n}\frac 1i}=\Omega(\log n). \\
    \end{align*}
    It remains to prove \eqref{equ:surrogate}. For each cell $j$, let $T_j$ be a random variable, indicating the time at which $j$ is filled. Fix $1\le i\le n$, and consider a random index $1\le j\le n-1$. After the $(n-i)$-th insertion, $i$ empty cells remain, so with probability $\Omega(i/n)$, either the $j$-th or the $(j+1)$-th cell is empty. Formally,
    \begin{align}
        \label{equ:hit_empty}
        \Pr_{\{x_n\},j}\Bk*{(T_j>n-i)\lor (T_{j+1}>n-i)}=\Omega\bk*{\frac in}.
    \end{align}
    Next, we show that conditioned on $(T_j>n-i)\lor (T_{j+1}>n-i)$, the difference between $a_j$ and $a_{j+1}$ is large:
    \begin{align}
        \label{equ:surrogate2}
        \Pr_{\{x_n\},j}\Bk*{\big|a_{j+1}-a_j\big|\ge \frac 1{4i}\biggr|(T_j>n-i)\lor (T_{j+1}>n-i)}=\Omega(1),
    \end{align}
    which, combined with \eqref{equ:hit_empty}, implies \eqref{equ:surrogate}.

    To prove \eqref{equ:surrogate2}, we consider three mutually exclusive events that collectively exhaust $(T_j>n-i)\lor (T_{j+1}>n-i)$:
    \begin{enumerate}
        \item Only the $j$-th cell is filled at time $n-i$.
        \item Only the $(j+1)$-th cell is filled at time $n-i$.
        \item Both cells are empty at time $n-i$.
    \end{enumerate}
    and we show that, conditioned on any of these events (if it has nonzero probability), the probability that $|a_{j+1}-a_j|\ge 1/(4i)$ is large.
    
    We first consider the first event. In this case, the value of the $(j+1)$-th cell $a_{j+1}$ must come from some future item $x_{k}(k>n-i)$. If all future items $x_{>n-i}$ are far away from the value of the $j$-th cell $a_j$, then regardless of how the algorithm chooses, we must have that the difference $|a_{j+1}-a_j|$ is large. Formally,
    \begin{align*}
        &\Pr_{\{x_n\},j}\Bk*{|a_{j+1}-a_j|\ge \frac 1{4i}\biggr|(T_j\le n-i)\land (T_{j+1}>n-i)} \\
        \ge {}&\Pr_{\{x_n\},j}\Bk*{\forall k\in (n-i,n],|x_k-a_j|\ge \frac 1{4i}\biggr|(T_j\le n-i)\land (T_{j+1}>n-i)} \\
        \ge {}&(1-\frac 2{4i})^{i}=\Omega(1),
    \end{align*}
    where the last line is due to the fact that, the value of the last $i$ items are i.i.d. random, and are independent from the event $(T_j\le n-i)\land (T_{j+1}>n-i)$ and the value of $a_j$. A symmetric argument also applies to the second event.
    
    The third event where both cells are empty is slightly different. In this case, let $k=\min(T_j,T_{j+1})>n-i$ be the first time that one of the cells is filled. We show that, with constant probability, future items $x_{>k}$ are all far away from $x_k$, which implies that $|a_{j+1}-a_j|$ is large. That is, for any $k>n-i$, if $\Pr[\min(T_j,T_{j+1})=k]>0$, then
    \begin{align*}
        &\Pr_{\{x_n\},j}\Bk*{|a_{j+1}-a_j|\ge \frac 1{4i}\biggr|\min(T_j,T_{j+1})=k} \\
        \ge{}&\Pr_{\{x_n\},j}\Bk*{\forall k'>k,|x_{k'}-x_k|\ge \frac 1{4i}\biggr|\min(T_j,T_{j+1})=k}\\
        ={}&\bk*{1-\frac 2{4i}}^{n-k}> \bk*{1-\frac 2{4i}}^{i}=\Omega(1),
    \end{align*}
    where we used the fact that conditioned on the event $\min(T_j,T_{j+1})=k$, the items $x_{>k}$ are randomly distributed. Formally, conditioned on the third event, we can show that
    \begin{align*}
        &\Pr_{\{x_n\},j}\Bk*{|a_{j+1}-a_j|\ge \frac 1{4i}\biggr|(T_j> n-i)\land (T_{j+1}>n-i)} \\
        &=\sum_{k>n-i,\Pr[\min(T_j,T_{j+1})=k]>0}\Pr_{\{x_n\},j}\Bk*{|a_{j+1}-a_j|\ge \frac 1{4i}\biggr|\min(T_j,T_{j+1})=k}\\
        &\;\;\;\;\;\;\;\;\;\;\;\;\;\;\;\;\;\;\;\;\;\;\;\;\;\;\;\;\;\;\;\;\;\;\;\;\;\;\;\;\;\;\cdot\Pr_{\{x_n\},j}\Bk*{\min(T_j,T_{j+1})=k\biggr|(T_j> n-i)\land (T_{j+1}>n-i)} \\
        &=\sum_{k>n-i,\Pr[\min(T_j,T_{j+1})=k]>0}\Omega(1)\cdot \Pr\Bk*{\min(T_j,T_{j+1})=k\biggr|(T_j> n-i)\land (T_{j+1}>n-i)} \\
        &=\Omega(1).
    \end{align*}
    This concludes the proof of \eqref{equ:surrogate2}, and thus the lower bound.
\end{proof}

\section{Open Questions}
\label{sec:open_problems}

In this paper, we present an upper bound of $\log n\cdot 2^{O(\log^* n)}$, as well as a lower bound of $\Omega(\log n)$. It would be interesting to close the gap between the two bounds. In particular, as the barrier to further improvements to the algorithm seems to be due to some inherent drawbacks of Chernoff bounds, we believe that it is possible to directly address this drawback, and prove a lower bound of $\omega(\log n)$.

Another interesting open question is, what if the number of cells is more than the number of items? \cite{Abrahamsen2024OnlineSA} presented an algorithm for stochastic online sorting that, when the number of cells is $\gamma\cdot n$ for $\gamma >1$, achieves expected cost $O(1+1/(\gamma-1))$. In comparison, it is not clear whether the techniques in this paper can benefit from an increased number of cells.

\section{Acknowledgement}

The author is thankful to William Kuszmaul for introducing the problem of stochastic online sorting, and for many useful comments on the early drafts of this paper. The author would also like to thank Erik Demaine, Jingxun Liang and Renfei Zhou for valuable discussions. We also thank the anonymous reviewers for helpful comments.

\bibliographystyle{alpha}
\bibliography{reference.bib}

\appendix

\section{Missing Proofs}
\label{app:warm_up_2}

In this appendix, we prove \cref{lem:warm_up_2_damp_order}, \cref{lem:warm_up_2_regularity} and \cref{lem:merger_concentration}. We can assume wlog that $n$ is sufficiently large.

\AdaptMergeSuccess*

\dampOrder*

We first define a concentration property that holds with high probability over the input $x_1,\dots,x_n$, then show that \cref{lem:warm_up_2_damp_order,lem:warm_up_2_regularity} hold deterministically conditioned on this property. \cref{lem:merger_concentration} is a technical lemma that will be used in the proofs.

\begin{property}
    \label{prop:warm_up_2_concentration}
    Let $n,K$ be as defined in $\algAdaptMerge$. Fix an integer $0\le \ell\le \lfloor\log_K n\rfloor$ and a segment $S=[(j-1)/K^\ell,j/K^\ell]$ of length $1/K^{\ell}$, where $j\in [K^{\ell}]$. Partition $S$ into $K$ sub-segments $S_1,\dots,S_K$, each of length $1/K^{\ell+1}$. Fix a number $1\le i\le n$. The following hold:

    \begin{enumerate}
        \item If $i\ge K^{\ell}\cdot \log^4 n$, then among the last $i$ input items $x_{n-i+1},\dots,x_n$, the number of items from $S$ is in the range
        \begin{align*}
            \frac{i}{K^{\ell}}\pm \frac {1}{20}\cdot d\bk*{\frac i{K^{\ell}}}.
        \end{align*}
        \item For any $m\ge K\cdot \log^4 n$, if at least $m$ items in $x_{n-i+1},\dots,x_n$ are from $S$: Let $y_{1\dots m}$ be the $m$ earliest such items. For any sub-segment $S_k$, the number of items in $y_{1},\dots,y_m$ that are from $S_k$ is in the range
        \begin{align*}
            \frac{m}{K}\pm \frac {1}{20}\cdot d\bk*{\frac mK}.
        \end{align*}
    \end{enumerate}
\end{property}

By a union bound over all the bad events, we can show that

\begin{claim}
    \cref{prop:warm_up_2_concentration} holds with probability $1-O(1/n^{10})$.
\end{claim}

\begin{proof}
    For the first set of properties, a direct Chernoff bound suffices, much like in the proof of \cref{clm:warm_up_phase_concentration}. For the second set of properties: Fix such a segment $S$, moment $i$ and number $m$, then conditioned on the fact that at least $m$ items in $x_{i \dots n}$ are from $S$, the distribution ofthe first such items $y_1,\dots,y_m$ is random. Also, $m$ is sufficiently large. Therefore, we can apply a Chernoff bound to bound the number of items from each $S_k$.
\end{proof}

Next, we show that conditioned on \cref{prop:warm_up_2_concentration}, both \cref{lem:warm_up_2_damp_order,lem:warm_up_2_regularity} hold. We start by the following claim:

\begin{claim}
    \label{clm:warm_up_2_damp_order_weak}
    Conditioned on \cref{prop:warm_up_2_concentration}, the following holds: When a buffer becomes full, if $\algAdaptMerge$ does not fail up to this point, then all the buffers that precede it are already full.
\end{claim}

We will later show that $\algAdaptMerge$ will not fail conditioned on \cref{prop:warm_up_2_concentration}, so that \cref{clm:warm_up_2_damp_order_weak} implies \cref{lem:warm_up_2_damp_order}.

To prove \cref{clm:warm_up_2_damp_order_weak}, we need to use \cref{lem:merger_concentration}, which shows some properties of the parameters used in merging.

\mergerConcentration*

\begin{proof}
    We expand the definitions of $m_j^{\text{old}}$ and $m_j$:
    \begin{align*}
        m_j^{\text{old}}&={\frac {n_j}{B_{j-1}}}-d\bk*{\frac {n_j}{B_{j-1}}}, \\
        m_j&={\frac {n_jK}{B_{j-1}}}-d\bk*{\frac {n_jK}{B_{j-1}}}.
    \end{align*}
    We can represent $m_j^{\text{old}}$ using $m_j$:
    \begin{align*}
        m_j^{\text{old}}&=\frac 1K\cdot \bk*{m_j+d\bk*{\frac {n_jK}{B_{j-1}}}}-d\bk*{\frac {n_j}{B_{j-1}}} \\
        &=\frac {m_j}{K}+\frac 1K\cdot d\bk*{\frac {n_jK}{B_{j-1}}}-d\bk*{\frac {n_j}{B_{j-1}}}.
    \end{align*}
    The number of future items $n_j$ is at least the size of the pool during phase $j-1$, which is 
    \begin{align*}
        n_j&\ge n_{j-1}-B_{j-1}\cdot m_{j-1} \\
        &=B_{j-1}\cdot d(n_{j-1}/B_{j-1}) \\
        &=\Omega(B_{j-1}\cdot(n_{j-1}/B_{j-1})^{2/3}).
    \end{align*}
    We also have that $n_{j-1}/B_{j-1}\ge K$, which is proven in \cref{lem:warm_up_2_decrement}. Therefore, $n_j/B_{j-1}=\Omega((n_{j-1}/B_{j-1})^{2/3})=\omega(1)$, which implies that for sufficiently large $n$, we have that $\frac 1K\cdot d\bk*{\frac {n_jK}{B_{j-1}}}=o( d(\frac {n_j}{B_{j-1}}))$ (since $d(x)=\Theta(x^{2/3})$). We now have that
    \begin{align*}
        m_j^{\text{old}}=\frac {m_j}{K}-(1-o(1))\cdot d\bk*{\frac {n_j}{B_{j-1}}}.
    \end{align*}
    Since $(n_j/B_{j-1})=(1+o(1))\cdot (m_j/K)$, the first bullet holds.

    To show that $m_j^{\text{old}}=\Omega(K^{2/3})$, we use the bounds mentioned before:
    \begin{align*}
        m_j^{\text{old}}&=(n_j/B_{j-1})-d(n_j/B_{j-1})\tag{By definition} \\
        &\ge (1/2)\cdot (n_j/B_{j-1})\tag{$d(x)\le(1/2)x$} \\
        &\ge (1/2)\cdot d(n_{j-1}/B_{j-1}) \tag{Lower bound on $n_j$} \\
        &\ge (1/4)\cdot((n_{j-1}/B_{j-1})^{2/3}-1) \\
        &\ge (1/4)\cdot(K^{2/3}-1) \\
        &\ge (1/5)K^{2/3} \tag{$n$ (hence $K$) is sufficiently large}.
    \end{align*}
\end{proof}

Now we can prove \cref{clm:warm_up_2_damp_order_weak}.

\begin{proof}[Proof of \cref{clm:warm_up_2_damp_order_weak}]
    The proof is by induction on the buffers. Fix a buffer $\buf$, and suppose that the claim holds for any buffer that precedes it.
    \begin{itemize}
        \item If $\buf$ is a regular buffer, then let $\buf'$ be the youngest buffer that precedes it (if no buffer precedes $\buf$, then the claim already holds for $\buf$). Due to the execution of $\algAdaptMerge$, since $\buf$ is a regular buffer, the initial segment of $\buf'$ must be the same as $\buf$. Therefore, when an item is inserted into $\buf$ (which happens before $\buf$ becomes full), it must be that $\buf'$ is full, which by the induction hypothesis implies that all the buffers that precede $\buf$ are full.
        
        \item If $\buf$ is a dampening buffer: Suppose that $\buf$ is allocated at the start of phase $j$. Then the buffers that precede $\buf$ can be divided into $K$ groups $\mathcal{B}_1,\dots,\mathcal{B}_K$, each corresponding to a phase-$(j-1)$ segment contained in the initial segment of $\buf$. Let $\buf_1,\dots,\buf_K$ be the youngest buffers of these groups, respectively. Similar to before, we have that, due to the execution of $\algAdaptMerge$, the initial segment of $\buf_i$ must exactly be the phase-$(j-1)$ segment of its group.

        Recall that, right after the allocation of $\buf$, each group $\mathcal{B}_i$ of buffers has total remaining capacity $m_j^{\text{old}}$, and $\buf$ has remaining capacity $m_j-Km_{j}^{\text{old}}$. Using \cref{prop:warm_up_2_concentration}, after allocating $\buf$, there must be at least $m_j=(n_j/B_j)-d(n_j/B_j)$ future items from the initial segment of $\buf$. Let $y_1,\dots,y_{m_j-1} $ be the first $m_j-1$ such items. Using \cref{lem:merger_concentration} and \cref{prop:warm_up_2_concentration}, we have that among $y_1,\dots,y_{m_j-1}$, each phase-$(j-1)$ segment has at least $m_j^{\text{old}}$ items. Using an argument similar to the first bullet, for the $i$-th phase-$(j-1)$ segment, after processing $m_j^{\text{old}}$ items from it, all the buffers preceding $\buf_i$ (including $\buf_i$) will be full. Therefore, we have shown that: After processing $y_1,\dots,y_{m_j-1}$, all the buffers that precede $\buf$ are full. By counting the number of items, we have that $\buf$ cannot be full at this moment. This concludes the proof. \qedhere
    \end{itemize}
\end{proof}

\cref{clm:warm_up_2_damp_order_weak} has the following useful corollary, which is a slightly different version of \cref{cor:warm_up_2_buffer_property}.

\begin{claim}
    \label{clm:warm_up_2_buffer_regularity}
    Conditioned on \cref{prop:warm_up_2_concentration}, the following holds: When a new item $x$ from current segment $i$ arrives, if the algorithm does not fail up to this point, and $c_i>0$, then $x$ will be inserted into a buffer belonging to $i$ (without failing).
\end{claim}

\begin{proof}
    Let $\buf$ be the current youngest buffer belonging to $i$. By the execution of our algorithm, the initial segment of $\buf$ is precisely $i$. Therefore, any other buffer belonging to $i$ precedes $\buf$. Since $c_i>0$, there exist buffers belonging to $i$ that are not full. Therefore, by applying \cref{clm:warm_up_2_damp_order_weak}, we have that $\buf$ is also not full. Therefore, $x$ will be inserted into a buffer, because if all else fails, it can always be inserted into $\buf$.
\end{proof}

Finally, we prove \cref{lem:warm_up_2_regularity} using \cref{clm:warm_up_2_buffer_regularity} and \cref{lem:warm_up_regularity}(the analysis for $\algAdapt$).

\begin{claim}
    Conditioned on \cref{prop:warm_up_2_concentration}, $\algAdaptMerge$ does not fail.
\end{claim}

\begin{proof}
    The proof is by induction: Suppose that the algorithm successfully reaches the start of some phase $j$. We show that it will also reach the start of phase $j+1$ (or terminate if $j$ is the last phase) without failing. To show this, we will largely reuse the analysis for $\algAdapt$. But first, we sort out the differences between $\algAdaptMerge$ and $\algAdapt$. There are two main differences between the failure criteria:

    \begin{itemize}
        \item Failure 2(b) does not exist in $\algAdapt$. For this, we use the first bullet of \cref{lem:merger_concentration}, which shows that the dampening buffers always have positive sizes ($m_j> Km_j^{\text{old}}$). That is, failure 2(b) will not happen for sufficiently large $n$.
        \item The condition for failure 3 is slightly different, in that in $\algAdaptMerge$, when a new item arrives, it may happen that there exist non-full buffers belonging to its segment, but none of them accepts the new item. However,  \cref{clm:warm_up_2_buffer_regularity} proves that we do not have to worry about this case.
    \end{itemize}
    With the two differences sorted out, the conditions of $\algAdaptMerge$ failing during phase $j$ are exactly the same as that of a phase in $\algAdapt$.
    
    In phase $j$, we start with $B_j$ segments, each having remaining capacity $m_j$. We also have $n_j/B_j\ge K$, i.e., the number of future items is large. Therefore, we can reuse the proof of \cref{lem:warm_up_regularity} to show that $\algAdaptMerge$ will not fail during phase $j$ (the details are omitted).
\end{proof}

\end{document}